%% file: main.tex
\documentclass[acmsmall]{amsart}
\usepackage{amssymb,amsmath,amsfonts}
\usepackage{amsthm}
\usepackage{paralist}
\usepackage{epsfig,xspace}
\usepackage{listings}
\usepackage{latexsym}
\usepackage{wasysym}
\usepackage[section]{algorithm}
\usepackage{algorithmic}
\usepackage{color}
\usepackage[scaled]{beramono}
\usepackage{setspace}

\title{Embedded Finite Models Beyond Restricted Quantifier Collapse}

\author{Michael Benedikt}
\address{Department of Computer Science, University of Oxford}
\email{michael.benedikt@cs.ox.ac.uk}

\author{Ehud Hrushovski}
\address{Mathematical Institute, University of Oxford}
\email{ehud.hrushovski.maths.ox.ac.uk}

\input{macros}

\begin{document}

 \input{abstract}
\maketitle


\input{intro}

\input{organ}

\input{def}

\input{higherbounded}

\input{separate}

\input{highcomp}

\input{robust}

\input{pseudofinite}

\input{rqcvsonebounded}

\input{conc}

\input{ack}

\bibliographystyle{IEEEtran}
\bibliography{main}

\newpage
\appendix
\input{appendix}

\end{document}

%% file: macros.tex
\newtheorem{proposition}{Proposition}
\newtheorem{theorem}{Theorem}
\newtheorem{example}{Example}
\newtheorem{corollary}{Corollary}
\newtheorem{claim}{Claim}

\newtheorem{lemma}{Lemma}

\usepackage{amsmath}
\usepackage{amsxtra}

\usepackage{enumitem}
\setlist[enumerate]{leftmargin=*}
  
\usepackage{amscd}
\usepackage{amsthm}
\usepackage{hyperref}
\usepackage{xspace}

\usepackage{thm-restate}

\usepackage[T1]{fontenc}

 \newtheorem{definition}{Definition}
\newcommand{\isoinv}{\kw{IsoInv}} 
\newcommand{\contained}{\subseteq}
\newcommand{\isunion}{\bigcup}
\newcommand{\domain}{\kw{Dom}}
\newcommand{\fraisse}{Fra\"iss\'e }
\newcommand{\powerset}{\mathcal{P}}
\newcommand{\myeat}[1]{}

\newcommand{\pspace}{\kw{PSPACE}}
\newcommand{\np}{\kw{NP}}
\newcommand{\alba}{\kw{ALessBA}}

\newcommand{\myparagraph}[1]{{\bf #1.}}
\newcommand{\kw}[1]{{\mathsf{#1}}\xspace}
\newcommand{\mymin}{\kw{Min}}
\newcommand{\mymax}{\kw{Max}}

\newcommand{\bilin}{{b}}

\usepackage{amsfonts}
\usepackage{amssymb}
\usepackage{eucal}
\usepackage[all]{xypic}
\usepackage{color} 

\usepackage{yhmath}
\DeclareSymbolFont{largesymbols}{OMX}{yhex}{m}{n}
\DeclareMathAccent{\widetilde}{\mathord}{largesymbols}{"65}

\newcommand{\adom}{\kw{Adom}}

\newcommand{\nc}{\newcommand}
\nc{\renc}{\renewcommand}
\nc{\ssec}{\subsection}
\nc{\sssec}{\subsubsection} 
\nc\ol{\overline}
\nc\wT{\widetilde{T}}
\nc\wh{\widehat}

\nc{\Pp}{{\mathbb{P}}}
\nc{\Rr}{{\mathbb{R}}}
\nc{\BV}{{\mathbb{V}}}
\nc{\BW}{{\mathbb{W}}}
\nc{\Zz}{{\mathbb{Z}}}
\nc{\Qq}{{\mathbb{Q}}}

 \nc{\Nn}{{\mathbb{N}}}

\newcommand{\nats}{\mathbb{N}}

\nc{\res}{{\mathop{\operatorname{\rm res}}}}

\def\meet{\cap}
\def\union{\cup}

\def\si{\sigma}

\def\m{\smallsetminus}

\nc{\seq}[1]{\stackrel{#1}{\sim}}

\def\inv{^{-1}}

\nc{\Ff}{{\mathbb{F}}}

\def\eeq{\end{equation}}

\def\prf{\begin{proof}}
\def\pv{\end{proof} }
 \def\eprf{\end{proof} }

 \renc{\b}{{\beta}}

\setlist[itemize]{leftmargin=*}

%% file: abstract.tex
\begin{abstract}
We revisit evaluation of logical formulas that allow both uninterpreted relations, constrained to be finite,
as well as an interpreted vocabulary over an infinite domain. This formalism was denoted \emph{embedded finite model theory} in the past.

It is clear that the expressiveness and evaluating complexity of formulas of this type depends heavily on the infinite structure. If we embed in a wild structure like the integers with additive and multiplicative arithmetic, logic is extremely expressive and formulas are impossible to evaluate. On the other hand, for some well-known decidable structures, the expressiveness and evaluating complexity are similar to the situation without the additional infrastructure. The latter phenomenon was formalized via the notion of ``Restricted Quantifier Collapse'': adding quantification over the infinite structure does not add expressiveness.
Beyond these two extremes little was known.

In this work we show that the possibilities for expressiveness and complexity are much wider.  We show that we can get almost any possible complexity of evaluation while staying within a decidable structure. We also show that in some decidable structures, there is a disconnect between expressiveness of the logic and complexity, in that we cannot eliminate quantification over the structure, but this is not due to an ability to embed complex relational computation in the logic. 

 We show failure of collapse for the theory of finite fields and the related theory of pseudo-finite fields, which will involve coding computation in the logic. As a by-product of this, we establish new lower-bounds for the complexity of decision procedures for several  decidable
theories  of fields, including the theory of finite fields.
 
In the process of investigating this landscape, we investigate several weakenings of collapse,
one allowing higher-order quantification over the finite structure, another allowing expansion
of the theory. We also provide results comparing collapse for unary
signatures with general signatures, and new analyses of collapse for natural decidable theories.

\end{abstract}

%% file: intro.tex
\section{Introduction}
This work concerns evaluating logical formulas over an infinite structure. For example, our structure could be
the real ordered field, and we want to evaluate a sentence $\phi$ that mentions the symbols $+,*,<$, and some constants -- say
`` is one element of the set $\{ 1/2, 1/3, 2/7\}$ the square root of another element of the set''.
Many structures exist such that this problem is decidable -- the real field, the complex field, integer arithmetic. And the complexity
ave been studied for many decades -- in the real field case, see for example \cite{bpr}.  In decidable cases, the complexity is generally high:
for example doubly-exponential for the real field.

Rather than analyzing the complexity of an arbitrary sentence as a function of the size of its string representation, one can look at families of sentences parameterized by a finite
relation. Generalizing the example above, we could take a finite set $P$ and ask the statement ``is one element of $P$ the square root of another
element of $P$''. Obviously the truth value will depend upon $P$, and the complexity will be higher the more elements $P$ has.
We can thus analyze the complexity of such a \emph{data-parameterized family} in the size of $P$. Following common practice in computational logic,
we  can denote this as the ``data complexity of the family'', indicating that everything except $P$ is fixed, as opposed to the ``combined complexity'' above, where the structure of the formula varies.

Let us formalize this notion of parameterized family a bit more. We can consider an  arbitrary theory $T$  -- a collection of logical sentences -- in some
language $L$. And we consider formulas over the language $L_V$ expanding
$L$ by predicates in a finite relational signature $V$, with the intention
that the interpretation of the $V$
predicates ranges over finite subsets in a model of $T$. When we talk about an $L_V$ formula, we always mean
a standard first-order formula in this signature.  Based on the intended semantics mentioned above,
two  $L_V$ formulas
are said to be equivalent (modulo $T$) if they agree over all finite interpretations of $V$
in a model of $T$. In past work,
such finite  interpretations are referred to as \emph{embedded finite models} (for $T$): see
\cite{benediktlibkinominimal,benediktsurvey,libkinsurvey}.  In analogy, we refer
to the case where $V=\{P\}$, $P$ a single unary predicate, as an \emph{embedded finite subset}.
Our data-parameterized families referred to informally above are just formulas of $L_V$, and the data complexity corresponds
to the complexity of evaluating a fixed  $L_V$ formula as the size of the $V$ interpretations varies.
There are several ways of measuring ``the size of $V$''. It could be  just the number of tuples in the interpretation, counting each tuple as a cost of one. Or we could  restrict to interpretations where all elements can be presented as strings, and refer to the size of the string encoding.  We defer the formalization of this in the introductory discussion below.

Above we mentioned the high complexity of evaluating logical sentences on a structure, in terms of the size of a sentence.
In contrast, it was discovered in the 1990's that the data complexity of parameterized formulas $L_V$, for common decidable models, is often polynomial in the size of $V$.
The general technique for showing such tractability results is to rewrite the $L_V$ formula to eliminate unbounded quantification over
the infinite models. That is, we transform to 
a special kind of
$L_V$ formula, a first-order formula built up from $L$-formulas
with  quantifiers ranging over elements in $V$: we call these \emph{first-order restricted-quantifier formulas}
or $1$-RQ formulas for short. If we assume that we have eliminated quantifiers in the $L$ formulas, this gives a family of quantifier-free formulas as we vary $V$. Assuming that the model and the interpretations are restricted so that evaluating quantifier-free formulas is tractable, the data complexity of this family will be low.

Prior work has identified conditions in which all sentences of $L_V$ are
equivalent  to $1$-RQ ones, either over all embedded finite models for
a given theory $T$ \cite{3belgians,benediktsurvey,flumziegler,casanovasziegler,belegradekisolation}
or over a particular class of such models  \cite{baldben}. For example, \cite{benediktlibkinominimal}
shows such results for a class of theories containing the real field, \cite{belegradekisolation} extends to a class
containing Presburger arithmetic, while  \cite{flumziegler} gives results for 
a class that includes the complex field.  These are sometimes
referred to as \emph{restricted-quantifier collapse} (RQC) results in the literature. They  show that the additional power
of quantification over an infinite structure in $L_V$ formulas can be ``compiled away''. In the presence of RQC,
  $L_V$ sentences are, in some sense, no more expressive than traditional first-order logic over finite models.
Much of the early work  was motivated from spatial databases, thus there was an emphasis on the case where the infinite universe
is the real numbers \cite{bpr,libkinsurvey}.
The above notion of RQC was for a theory, with the compilation producing an equivalent formula over all models of the theory. We can also say that a \emph{structure} has RQC if every $L_V$ formula is equivalent to a $1$-RQ one over the structure. This
is equivalent to the theory of the structure having RQC.

Previous research focused on the case where RQC holds, and its consequences for the complexity of evaluation.
It is not clear what happens when RQC fails. What can we say about the complexity or the expressiveness of $L_V$ formulas
in this case?  What are the possibilities for the expressiveness of $L_V$ formulas over all models, or over all decidable models?
In this paper we make a step towards a more complete picture, looking at models that do not have RQC, and 
 investigating several weakenings of the notion of 
RQC.

\myparagraph{Higher-order collapse and higher complexity theories}
Our goal is to explore theories which behave worse than RQC, but where we still get some control over expressiveness and complexity of evaluation.
Towards that end, we consider a weakening of RQC where we eliminate unrestricted quantification in favor
of \emph{higher-order quantification} over the finite structure.  
We say that a theory is $k$-RQC if every
$L_V$ formula is equivalent (over all embedded  finite models for $T$) to 
a formula built up from $L$-formulas using $i^{th}$ order quantification over elements in $V$
for $i \leq k$.  Restricted-quantifier collapse from the literature corresponds to $1$-RQC.
We call a theory $\omega$-RQC if every formula can be converted
to a $k^{th}$-order one for some $k$, but not necessarily with a uniform bound on $k$.
It is easy to see that completely computationally badly-behaved structures
like  full integer arithmetic cannot have this property -- indeed, $\omega$-RQC implies that one cannot use first-order formulas to perform a computation on the finite part of the structure that is non-elementary in the size of the structure alone.
We provide examples that show  that the hierarchy of theories with $k$-RQC does not
collapse. Our construction will further show that the complexity of evaluating logical formulas can be as wild as possible.


\myparagraph{Disconnects between collapse and complexity} In our investigation of higher-order collapse, we provide a recipe for generating decidable theories where first-order logic over the ambient structure can capture higher-order logic over the embedded finite structure, and examples where the complexity of evaluation is arbitrarily high. But some theories have a disconnect between expressiveness of the logic and complexity: there are first-order sentences over embedded finite models  which we cannot express in any reasonable logic quantifying over the finite structure alone, but the complexity of evaluation in the theory is still not very high. We give examples where this disconnect can be ``fixed'', by expanding the signature, so that the level of collapse reflects the complexity of evaluation. But we give a natural example, using Boolean Algebras, where this disconnect is fundamental:  the complexity of evaluation is elementary, but we cannot obtain RQC even by extending the signature.

\myparagraph{Notions of collapse and complexity}
Although many of our main results on  weaker notions of collapse are negative -- certain theories do not admit this form of collapse --
we show that our results imply  lower bounds on complexity
of decision procedures for the underlying theory.  

\myparagraph{Evaluating formulas over finite fields} We turn to the question of whether there are  important and natural decidable theories where both the expressiveness of the logic and also the complexity of evaluation is high. We demonstrate this for  the theory of finite fields and for the theory of
pseudo-finite fields \cite{ax}. The former refers to the first-order sentences
which hold in every finite field, while the latter adds on the fact that the model is infinite.
We show that the theory of finite fields and the theory of pseudo-finite fields is  not $\omega$-RQC.
 We use results about (failure of) RQC for  pseudo-fields to give lower bounds
on quantifier-elimination for the theory. We hope this connection between collapse to
restricted-quantifier formulas
and lower bounds can be exploited for other examples.

\myparagraph{Restricting the signatures for the uninterpreted predicates}
Finally, we consider weakening collapse by requiring that it holds only for signatures $L_P$: that is, for embedded finite subsets.
We show that this weakening of $1$-RQC is actually not a weakening at all: it is equivalent to $1$-RQC.
We show that for a variation of RQC studied in past work -- focusing on $L_V$ formulas depending only on the $V$ isomorphism type --
 the arity of relations in $V$ can make a difference.
 
\myparagraph{Our techniques}
The prior results in embedded finite model theory \cite{benediktsurvey,belegradekisolation} use
classical model-theoretic techniques. For example, results that unrestricted quantification can
be eliminated from real closed fields rely on  some very basic  properties of fields, namely o-minimality.
We hope that our work also offers some insight into the use of more recent technical
tools. In investigating finite fields, we naturally make use of the techniques developed for the theory of finite fields over the last decades, including quantifier-elimination and expressiveness results \cite{ax,zoelouangus}. In our results on impact of the signature, we make use of recent results on NIP theories, an area that has developed
rapidly over the last decade:
\cite{simontwovarsip,simon-chernikov-two}. NIP theories have strong links with both machine learning and efficient model checking
\cite{szymonnip}. In investigating weaker notions
of collapse one main tool is a construction of Henson
\cite{henson72}, which has also found other applications in theoretical computer science recently \cite{bodirskyhenson}.
In investigating weakening of collapse via expansion, we make use of results on Ramsey expansions of theories
\cite{definabilitypatterns}.

%% file: organ.tex
\myparagraph{Organization} Section \ref{sec:def} reviews the basic definitions 
around embedded finite model theory, and also overviews some older results on the topic.
Section \ref{sec:higher} studies higher order collapse.
Section \ref{sec:robust} investigates whether the failure of collapse
can be fixed by extending the signature.
Section \ref{sec:pseudo} focuses on the case of finite and pseudo-finite fields.
Section \ref{sec:arity} presents our results on collapse in monadic relational signatures vs collapse
over all relational signatures. 
We close with a  discussion of our  results and open questions in Section \ref{sec:conc}.

The body of the paper focuses on results of interest to theoretical computer science.
Some further algebraic examples, some background on a quantifier-elimination result used in the finite field section, along with  supplementary results of a more model-theoretic
nature,   are included the appendix.

%% file: def.tex
\section{Definitions and prior results} \label{sec:def}

Let $T$ be a complete theory in a language $L$.
Fix a finite relational signature  $V$ disjoint from $L$
and let $L_V$ be the union of $L$ by $V$. We write an $L_V$ structure as $(M, I)$.
Given a $V$ structure $I$ , the 
\emph{active domain} of $I$, $\adom(I)$,  is the union of the projections of the interpretations
of symbols in $V$.  Thus if $I$ interprets each relation in $V$ by a finite set of tuples,
$\adom(I)$ is a finite set.
An \emph{$L_V$ formula} always means an ordinary first-order formula in this
signature.
A \emph{first-order restricted-quantifier formula}
(or just ``$1$-RQ formula'') is built up from first-order $L$-formulas and
$V$ atoms by quantifications of the form $\exists x \in \adom~ \phi$ or
$\forall x  \in \adom  ~ \phi$ where $x$ is a variable. The semantics on $L_V$ structure
$(M,I)$ with valuation $\sigma$   is that $\exists x \in \adom~ \phi$ holds when there is $x_0$ in $\adom(I)$ 
such that $\phi$ holds on $\sigma$ extended with $x \mapsto x_0$. The semantics of $\forall x \in \adom ~ \phi$ is given similarly, or by duality.
It is easy to see that these formulas can be translated to special kinds of $L_V$ formulas: we can expand out
$\exists x \in \adom ~ \phi$ into  disjunctions of formulas having the form $\exists x \vec y ~ A(x, \vec y) \wedge \phi$.

We will also focus on the case $V = \{P\}$, abbreviating $L_V$ as $L_P$.
Here active domain quantification is just quantification over $P$, and we write
 $\exists x \in P ~ \phi$ or $\forall x \in P ~ \phi$ for such quantification. We do not need atom
$P(y)$ in the base case of the syntax, since they can be mimiced by quantifications
$\exists x \in P ~ x=y$.

\input{examplebounded}

Note that in the prior literature many other names are used for these formulas: e.g. ``active domain''
or ``restricted''. In the literature, attention is often focused on
theories with quantifier-elimination, and RQ  formulas are built up only from \emph{atomic
$L$ formulas} rather than arbitrary $L$ formulas. In this work we will not assume quantifier-elimination
in the theory, and thus use the more general definition.

An \emph{embedded finite model} for theory $T$ is a pair $(M,I)$ as above where $M \models T$, and
$\adom(I)$ is a finite subset of the domain of $M$.
An \emph{embedded finite subset} for theory $T$ is just the special case where $V=\{P\}$ with $P$ unary:
A pair $(M, P)$ where $M, \models T$ and $P$ is a finite subset of a domain.
We say two formulas $\phi(\vec x)$ and $\phi'(\vec x)$ of $L_V$
are \emph{equivalent over $T$} if   in every embedded finite model  $(M,I)$ for $T$,
$(M, I) \models \forall \vec x  ~ [\phi \leftrightarrow \phi']$.

We now come to our central definition:
\begin{definition} \label{def:onerqc}
We say that a theory $T$ has $1$-Restricted Quantifier Collapse (``is $1$-RQC'')  if:
 every $L_V$ formula  is equivalent to a $1$-RQ formula.
\end{definition}
In the prior literature $1$-RQC is referred to just as RQC \cite{benediktsurvey}: the 
reason for the prefix ``$1$'' will be clear when we introduce higher-order generalizations below.

\subsection{A little bit of model theory}
We briefly overview some of the model theoretic notions that are relevant to our work.

\myparagraph{Indiscernible sequences}
Let $J$ be a linear-ordered set and $a_i: i \in J$ be a $J$-indexed sequence in an $L$ structure $M$:
an injective function from $J$ into an $L$ structure $M$.
The sequence is \emph{order indiscernible} (or
order $L$-indiscernible, if $L$ is not clear from context) if
for any $k$,  any two $k$ tuples $a_{n_1}  \ldots a_{n_k}$, $a_{n'_1} \ldots a_{n'_k}$ with $n_1 < n_2 \ldots < n_k$
and $n'_1 <n'_2 \ldots <n'_k$,  and any $L$ formula
$\phi(x_1 \ldots x_k)$, $M \models \phi(a_{n_1} \ldots a_{n_k}) \leftrightarrow \phi(a_{n'_1} \ldots a_{n'_k})$. That is, the ordering
on the indices determines the formulas satisfied by elements in the sequence. We will often deal with the case where
the elements are ordered $a_i: i \in \nats$, and refer to this as an \emph{indiscernible sequence}.
Every theory with an infinite model has one
 with an
 indiscernible sequence \cite{hodgesbook}. A \emph{totally indiscernible set} is an infinite
set $A$ in a model where $M \models \phi(\vec a) \leftrightarrow \phi(\vec a')$ for each
$k$-tuple of distinct elements $\vec a, \vec a'$ and each $L$-formula $\phi(x_1 \ldots x_k)$. We often just
talk about an indiscernible set, where the context makes clear whether it is order indiscernible for some ordering
or totally indiscernible. If $\vec d$ is a subset of an $L$-structure $M$, a sequence is
\emph{indiscernible over $\vec d$} if it is indiscernible in the model for the extension of $L$ with constants
interpreted by $\vec d$.

\myparagraph{NFCP theories} A theory $T$ is NFCP\footnote{NFCP stands for ``Not the Finite Cover Property''.  
The reader without a background in model theory will probably
find expanding acronyms of model-theoretic classes
is not insightful,  while for readers with such a background spelling out the acronym is unnecessary.
We thus avoid spelling out  for the other classes (e.g. NIP) below. }if
it satisfies a strong quantitative form of the compactness theorem.
For every $\phi(x_1 \ldots x_j, y_1 \ldots y_k)$ there
is number $n$ such that  for any finite set  $S$ of $k$-tuples in a model $M$ of $T$,
if for every subset $S_0$ of  $S$ of size at most $n$, $M \models \exists \vec x~
\bigwedge_{\vec s \in S_0} \phi(\vec x, \vec s)$, then
$M \models \exists \vec x \bigwedge_{\vec s \in S} \phi(\vec x, \vec s)$.
Examples of NFCP theories include the theory of algebraically closed fields in
each characteristic -- in particular, the theory of the complex numbers.

NFCP theories are inherently \emph{unordered}: they do not admit a definable linear order.
In fact, a much stronger statement is true: NFCP theories are \emph{stable}, informally meaning
that there is no formula $\phi(\vec x, \vec y)$, such that $\phi$ restricted to arbitrarily large 
finite sets of tuples defines a linear order.

It is known that there are very basic stable theories that are not $1$-RQC. One of them
is the canonical example of a theory without NFCP, which will be particularly relevant to our
discussion.
\begin{example} \label{ex:fcp}
Let $L=\{E(x,y)\}$ and $T$ be the $L$-theory stating that $E$ is an equivalence
relation with classes of each finite size. Consider the $L_P$ formula $\phi_{\contained}$ stating
``some equivalence class is contained in $P$''. It is easy
to show that this is not equivalent to any $1$-RQ formula. Informally
 with a $1$-RQ formula, all we can say about a finite set $P$ that lies within a single equivalence
class are Boolean combinations of cardinality bounds: $|P|\geq k$ for fixed $k$.
\end{example}

\myparagraph{O-minimal theories}
We now turn to a model theoretic tameness property for \emph{linear ordered} structures.
Consider a theory $T$ with a relation $<(x,y)$ such that $T$ implies that $<$ is a linear order.
Such a $T$ is \emph{o-minimal} if for every $\phi(x, \vec y)$ for every model $M$
of $T$ and any $\vec c$ in $M$, $\{x | M \models \phi(x, \vec c)\}$ is a finite
union of intervals.
The real ordered group, real ordered field,  and the real exponential field
are all o-minimal \cite{ominimal}. Since no NFCP theory has a definable linear order, NFCP and o-minimal are
disjoint.

\myparagraph{NIP theories}
The most important class we deal with here are NIP theories \cite{simonnip}.
Given  $\phi(x_1 \ldots x_j, y_1 \ldots y_k)$, and a finite set of $j$-tuples $S$ in a model
$M$ we say that  $S$ is \emph{shattered} by $\phi_{\vec y}$ if for each subset $S_0$ of $S$ there
is $\vec y_0$   such that $S_0=\{\vec s \in S| M \models \phi(\vec s, \vec y_0)\}$.
A theory $T$ is NIP if for each formula $\phi$ and each partition of the free variables into
$\vec x, \vec y$, there is a number that bounds the  size of a set
shattered by $\phi_{\vec y}$ in a model of $T$.
NIP theories include both ordered and highly unordered structures. Specifically, they contain
o-minimal structures, Presburger arithmetic, as well as all stable structures, and hence all NFCP structures.
NIP can also be rephrased in terms of the well-known notion of VC dimension in learning theory.
Every partitioned $L$ formula $\phi(\vec x, \vec y)$ gives a family of subsets of the $j$-tuples in a model, as we vary $\vec y$. NIP can be restated as asserting
that for every $\phi$, the corresponding family has finite VC dimension \cite{laskowski}.

\myparagraph{Sufficient conditions for $1$-RQC}
In the setting of ``unordered'' structures, the main known example of $1$-RQC comes from NFCP theories:
\begin{theorem}  \cite{flumziegler} 
Every NFCP theory is $1$-RQC.
\end{theorem}
Note that this includes the case of pure equality, which was known very early in database theory \cite{hullsunatact}.
And this result implies that the field of complex numbers is $1$-RQC.

On the side of ``ordered structures'', a
model-theoretic sufficient condition concerning $1$-RQC involves o-minimal theories:

\begin{theorem} \cite{benediktlibkinominimal} Every o-minimal theory is $1$-RQC.
\end{theorem}

\cite{belegradekisolation} showed $1$-RQC for an even broader class, what today
are known as \emph{distal theories}. We will not need the definition here,
but it subsumes $o$-minimal theories, Presburger arithmetic, and the theory
of the infinite tree, while being contained in NIP. Thus  in particular
\emph{the real closed field and Presburger arithmetic have $1$-RQC}.

An easy observation is that NIP is  necessary for $1$-RQC:

\begin{proposition} \label{prop:rcqnip} \cite{benediktsurvey} If $T$ is $1$-RQC, then $T$ is NIP.
\end{proposition}

We will strengthen this result in  Theorem \ref{thm:forall1boundednip} below.

We mentioned previously that Example \ref{ex:fcp} is stable, hence NIP. This is true for the theory of
any equivalence relation. Since Example \ref{ex:fcp} is not $1$-RQC, we see that NIP theories do  not necessarily have $1$-RQC.
However, results from \cite{baldben} show that in NIP theories $T$, we can eliminate unrestricted quantifiers for
$L_V$ sentences that are \emph{isomorphism-invariant}, where this means that \emph{their truth only depends on the
isomorphism type of the $V$ predicates}. 

\begin{theorem} \label{thm:bbiso} \cite{baldben} In every NIP theory, there is a model containing an infinite
set $I$ such that for every $\isoinv$ $L_V$ formula $\phi$, there is a $1$-RQ formula $\phi'$
such that $\phi$ and $\phi'$ are equivalent, over embedded finite structures in the model where the domain of
the $V$ structure lies  within $I$.
\end{theorem}
The result in \cite{baldben} is stated a bit differently: we explain how to bridge the gap between them in the appendix.

 This result limits what can be expressed by a sentence in $L_V$ whose truth value depends
only on the isomorphism type of the $V$ structure.
Isomorphism invariance allows us to assume that the embedded finite structure
is based on  a set of indiscernibles for the model, and on such a set we can replace all
$L$-formulas by
a linear order representing the order on indiscernibles. From this we can see that  $\exists$ $1$-RQC implies that isomorphism-invariant sentences are 
expressible in 
\emph{order-invariant first-order logic}. Syntactically, this is
 first-order logic over a relational signature $V$ augmented with
an additional symbol $<$, with the additional semantic requirement
that the result of $Q$ is the same whenever $<$ is interpreted  by a linear order on the domain of
the finite structure, regardless of how the elements are ordered. For more 
on the expressiveness of order-invariant FO on finite models, see   \cite{nicoleorderinvariant}.

Thus Theorem  \ref{thm:bbiso} and known limitations of FO on finite structures imply the following:

\begin{corollary} \label{cor:nipordinv} In an NIP theory any isomorphism-invariant $L_V$ sentence
is expressible in order-invariant FO over the $V$ vocabulary.
Hence no $L_P$ sentence can express  a property of $|P|$ that holds
on  infinitely and co-infinitely many cardinalities.
And no $L \cup \{G\}$ sentence can express that $G$ is connected. 
\end{corollary}

\myparagraph{\fraisse  Limits}
We will make use of a well-known model-theoretic construction known as the
 \fraisse  limit. We  apply it to
 a  class $C$ of isomorphism types of finite structures in a finite relational signature
that satisfy some closure properties: the Joint Embedding Property and the 
Hereditary Property. 
 Fra\"iss\`e's theorem
states that for any such class $C$ there is a countably infinite model such that  $C$ is the class of (isomorphism types) of finite substructures of $M$, and which is \emph{homogeneous}: in our setting, this means that every isomorphism between finite substructures of the model extends to an isomorphism
of the whole model.
Furthermore, if $C$ satisfies another closure property 
-- the Amalgamation Property --  then there is a unique such
structure.  One canonical
example  is where $C$ is just the class of graphs. In this case, the  \fraisse limit is the 
\emph{random graph}, whose theory is axiomatized by the \emph{extension axioms} stating that for every finite graph $G$ and
every supergraph $H$ of $G$, any embedding of $G$ in the model  extends  to an embedding of $H$.
Another examples is where $C$ is the class of finite linear orders, where the  \fraisse  limit is simply a countable dense linear order without endpoints. We will not need the precise
definitions of the closure properties above, but
see, e.g. \cite{hodgesbook}.


%% file: examplebounded.tex
\begin{example} \label{ex:basic}
Let $T$ be the theory of an infinite linear order $<$. If $V$ contains two unary predicates
$P$ and $Q$, then one  $L_{V}$ sentence states that there is a real number above every member of
$Q$ and below every member of $P$:
\[
\exists x ~ (\forall q \in Q ~ x>q) \wedge(\forall p \in P ~ p< x) 
\]
This sentence is not $1$-RQ. But it is equivalent to the following $1$-RQ sentence:
\[
 \forall p \in P ~ \forall q \in Q ~ p>q
\]
\end{example}

%% file: higherbounded.tex
\section{Higher-order collapse and high complexity evaluation} \label{sec:higher}
How can we obtain structures where evaluating a first-order logic formula over the model is more complicated than in $1$-RQC structures, but not completely wild as in integer arithmetic?

We will consider a first natural weakening of RQC by loosening the notion of restricted-quantifier formula, allowing higher-order quantification
over the active domain.
We  define higher-order sorts by starting with some base sort $B$ and closing
under tupling and power sets. The order of a sort is the maximum number of nested powersets:
e.g. a sort $\powerset(\powerset(B) \times B)$  has order $2$. The first
component is a set of elements of base sort and the second is an element of base sort.
Given a finite set $P_0$, and a variable  of sort $\tau$, the valid valuations of the variables are defined
in terms of the hierarchy over $P_0$: for a variable $X$ of sort  $\powerset(\powerset(B) \times B)$,
a valid valuation of $X$ assigns  to it sets of pairs, where the first
component is a set of elements of $P_0$ sort and the second is an element of $P_0$.

Given relational schema $V$ disjoint from
$L$, the \emph{higher-order restricted-quantifier formulas}  over $L \cup V$ are built up
from $L$ formulas, atomic $V$ formulas, and atoms $X(u_1 \ldots u_m)$ where $X$ and each
$u_i$ are higher-order variables. The inductive definition is formed
by closing under the (implicitly) restricted quantifications $\exists X$ and $\forall X$ for
$X$ a higher-order variable.   The semantics on embedded finite model $(M,I)$
is in terms of valuations  over $\adom(I)$. 
For $x$ of sort $B$ we write out $\exists x ~ \phi$ as $\exists x \in \adom ~ \phi$, in order to make the bounding
explicit in the syntax and to be consistent
with our definition of $1$-RQ formulas. Similarly for universal quantifiers restricted to the active domain.
 For $X$ of sort  $\powerset(B)$, we write 
$\exists X ~ \phi$ as $\exists X \subseteq \adom ~ \phi$, again to make the bounding explicit. And similarly for universal
quantification and at higher orders.
In the case where $V=\{P\}$, $P$ 
a unary predicate, we
can simplify the restricted existential second-order quantifier as
$\exists X \subseteq P ~ \phi$.
For a number $k$, a \emph{$k$-restricted-quantifier formula} or just $k$-RQ formula is one where the variables are of order
at most $k$.

\begin{example}
Let $L=G(x,y)$ and $V=\{P(x),Q(x)\}$.
If $X$ is a variable of sort $\powerset(B)$  then 
$\forall S \subseteq P  ~ \exists x\in Q ~ \forall y \in P ~(G(x,y) \leftrightarrow S(x))$ is a $2$-RQ formula. Informally, it says that using $G$ and elements
of $Q$, we can pick out any subset of $P$.
\end{example}

We are now ready to give our first weakening of $1$-RQC:
\begin{definition} \label{def:krqc}
A theory is \emph{$k$-RQC} if for every finite relational
signature $V$, every  $L_V$ sentence is equivalent
to a $k$-RQ one. 
\end{definition}
Higher-order RQC has not been explored in any depth. But
there is one simple observation in the literature:
(see \cite{libkinsurvey}).
\begin{proposition} The random graph is $2$-RQC but not $1$-RQC.
\end{proposition}

We can weaken collapse further by dropping the requirement of a uniform bound.
We say $T$ is $\omega$-RQC if for every $L_P$ formula  $\tau$ there is  a formula
$\theta$ that is  $k$-RQ for some $k$
such that $\theta$ is equivalent to $\tau$ over all embedded finite subsets.

It is easy to see that in considering $k$-RQC for any fixed
$k$, or $\omega$-RQC
for a theory, we can restrict attention to  any model of the  theory.
Although $\omega$-RQC has not been studied explicitly,
it is easy to see that certain examples that are not $1$-RQC are also 
not $\omega$-RQC:

\begin{proposition} The theory of full integer arithmetic $(N, +, \cdot)$ is
not $\omega$-RQC.
\end{proposition}

A key motivation for $k$-RQC comes from complexity.  In the context of logics over finite structures, there is a strong connection between definability in higher-order logics and complexity. For example existential second-order definability in a language with only a relational signature -- that is, no interpreted structure --  corresponds to the complexity class $\np$. While  full second-order logic, again in the setting with out interpeted structure,  corresponds to $\pspace$ \cite{immermanbook}. In the context of finite models, higher-order definability has a correspondence with the  exponential time hierarchy, see \cite{hellahigherorder}.
 In the presence of interpreted structure, we can see that if $L_V$ formulas can express arbitrary $k$-RQC sentences, then
 the complexity of evaluation is (at least) non-elementary -- just inheriting results from the case of finite structures. Such coarse lower bounds are independent of how the infinite structure is encoded: thus in this work, \emph{when talk about lower bounds on evaluating sentences,  they will be in this sense of isomorphism-invariant formulas}.
 
We cannot claim that $k$-RQC implies elementary upper bounds on the complexity of evaluation for arbitrary sentences. Complexity upper bounds depends
on how easy it is to evaluate quantifier-free sentences, which will depend on the specifics of the presentation of the model.
But we can see that:
\begin{proposition} \label{prop:higherorderelementary} If $T$ has $k$-RQC, then any isomorphism-invariant sentence $\phi$ of $L_V$ the set of isomorphism types of finite $V$ which satisfy $\phi$ can be recognized by a Turing Machine running in $k$-iterated exponential
time.
\end{proposition}
\begin{proof}
Let $I_\phi$ be the set of $V$ isomorphism types of $V$ structures that satisfy $\phi$ in some (equivalently) all embedded finite structure.  By $k$-RQC we can rewrite the sentence $\phi$ to $k$-RQ formula. There is a model of $T$ with a set of order-indiscernibles $I$, and
over $I$ the $k$-RQ formula can be rewritten to a $\phi'$ use only the ordering symbol. Let $\phi''$ be obtained
by existentially quantifying over the linear order. Given an isomorphism type, presented on a tape with an arbitrary ordering, $V_0 \in I_C$ if and only if $\phi'$ holds of it as an ordered structure. This can be checked in time a tower in $k$.
\end{proof}

A partial converse will be given in Proposition \ref{prop:elem}.

%% file: separate.tex
\subsection{Separating RQC levels and Higher-order Spectra} \label{subsec:separate}
We will show that for any $k$ there is a $T$  that
is $(k+2)$-RQC  but not $k$-RQC. We will also show that
there are decidable theories that are not even $\omega$-RQC.

As mentioned above, there is a strong connection between definability in higher order logics and complexity. For example existential second order definability in a language with only a relational signature -- that is, no interpreted structure, corresponds to the complexity class $\np$, while full second order logic corresponds to $\pspace$. Higher-order definability has a close connection to the the exponential hierarchy \cite{hellahigherorder}.

In the process of separating $(k+2)$-RQC and $k$-RQC,  \emph{we will obtain examples of decidable structures where the complexity of evaluating formulas can be
arbitrarily high.}

$k^{th}$-order logic is the logic built up using
variables  of order $1$ to $k$, with atomic formulas
equality as well as $X(x)$ where $X$ has order $j+1$ and $x$ order $j$ via Boolean connectives
and quantifiers $\exists S$. We will consider ``pure equality higher-order sentences'':
$k^{th}$ order logic sentences with no relational constants, only quantified
variables. We can interpret such sentences 
over a structure consisting only of a domain -- a pure set.
We call a set $J$ of numbers a \emph{pure $k^{th}$ order spectrum} if there is a sentence $\phi$ of the above form
such that  $J= \{n  ~ \mid  \{1 , \ldots n\} \models \phi\}$.

The following is a variation of results in
\cite{hellahigherorder}:
\begin{proposition}  \label{prop:spectrum}
For any $k$, there is a set $J_k$ of numbers that is a pure $k+2$ order spectrum 
but not  a pure $k$ order spectrum.
\end{proposition}

Note that for formulas with only equality, there is a one-to-one correspondence between models of the formula and integers, where the correspondence
maps a model to its size. Thus it is enough to separate the expressiveness of $k+2$ order logic from that of $k$ order logic, over a unary signature.
This last follows from results  in \cite{mostowskisemantics,leszek}, stating that a truth definition for a $k$ order logic  sentence can always
be expressed as $k+2$ order logic sentence. This holds for any fixed signature, hence also for a fixed signature.

Note that \cite{hellahigherorder,hullsuhigherorder} claim separation results for expressibility in higher-order logic,  following
from complexity characterizations of the definable sets in the logic: in \cite{hellahigherorder} the characterization is via non-deterministic iterated
exponential time with an oracle for a class in the polynomial hierarchy. Similar characterizations are provided in \cite{leszek}.
In \cite{hellahigherorder}, the argument that is sketched for separation does not apply to a signature with only equality
In \cite{hullsuhigherorder} the  argument for separation relies on unpublished results in a Phd thesis \cite{thesisspectrum}, so it is difficult to say the signatures to which it applies.
But it seems likely that the separation result does not apply to the case of a trivial signature.

We now show how to generate theories $T_J$ from a set of numbers $J$, such
that $T_J$ is  $k$-RQC if and only if $J$ is a pure $k$-order spectrum.
Thus separation of RQC levels corresponds to separation of 
higher-order definability in finite model theory.  
As a consequence we will have:

\begin{theorem} \label{thm:separate} For each $k$ there are decidable complete theories 
$T_k$ that are  $(k+2)$-RQC but not $k$-RQC.
\end{theorem}

The recipe for obtaining $T_J$ from $J$ can be used to reduce other separation
statements (e.g. about the variation of $k$-RQC where only Monadic higher order quantification is allowed) to 
separations in finite model theory.
We use a  construction due to Henson \cite{henson72}, first used to show that there are
uncountably many homogeneous directed graphs. See \cite{dugaldsurvey} for more background on these constructions.


 A graph is
\emph{decomposable} if there is a non-trivial
set $S$ of vertices with the property that for each pair
$x, y \in S$ and every other vertex $v$ in the graph, there is an
edge from $x$ to $v$ if and only there is an edge from $y$
to $v$  and similarly for edges from $v$ to $x$. That is, elements of $S$ behave
the same way with respect to elements outside of $S$. A graph is \emph{indecomposable} if
it is not decomposable.
What is achieved by the construction of Henson can be stated as follows:

\begin{theorem} \cite{henson72} \label{thm:henson}
For each natural number there is a finite 
indecomposable directed graph $B(n)$ such that for $n \neq n'$,  $B(n)$ does not embed as an induced substructure
of $B(n')$.  
 Furthermore,
there is a first-order interpretation $\delta$ whose input language is just a binary language $<$,
such that applying the interpretation to 
a linear order of size $n$ gives $B(n)$.  
\end{theorem}

A first-order interpretation of dimension $d$ \cite{hodgesbook} is a function from 
structures to structures, specified by  a formula $\phi_{\domain}(x_1 \ldots x_d)$ in the input
vocabulary describing the domain of the output, along with, for each $k$-ary relation
symbol in the output vocabulary, a $(k \cdot d)$-ary formula in the input vocabulary.
An interpretation allows a formula also for equality -- meaning that the domain of the output
is actually a quotient of the $k$-tuples  satisfying $\phi_{\domain}$. But such a quotient
is not needed in the interpretations of Theorem \ref{thm:henson}.

Let a \emph{cone} be a digraph of the form $\{d\} \times D$ or $D \times \{d\}$.
One can verify that the construction of \cite{henson72} has the property
that no   $B(n)$ is an induced substructure of a cone.

We consider a signature consisting of  ternary relations $R$
and  $Z$.
Given a structure $M$ for such a relation  and $a$ in the domain of $M$
we let $R_{a}$ be the binary relation such that $R_{a}(b,c)$ holds exactly when
$R(a, b,c)$ holds. Let $Z_a$ be a binary relation defined analogously.

For a set of numbers $J$, let $K_J$ be the set of finite structures for $R,Z$
such that $R(x,y,z)$ implies $x,y, z$ are all distinct and for no $a$ in the domain do we have:
 $B(n)$ for $n \in J$ embeds as an induced substructure of $R_a$.
This set of structures is easily seen to 
 have the hereditary, amalgamation, and joint embedding
properties. Thus there is a \fraisse limit (see Section \ref{sec:def})
which
we denote as  $M_J$. Let  $T_J$ be the complete theory of 
$M_J$.
Since there was no restriction on $Z$,
in $M_J$ the relation $Z$ will simply be a random ternary relation.
Whenever $J$ is a decidable set, the class of finite structures $K_J$
is effective, and thus the theory $T_J$ is computably axiomatizable via the usual ``extension axioms'',
stating that for every tuple $\vec t_1$ in the model
and isomorphism $h_1$ of a structure  $s_1$ in $K_J$ to the substructure of the model induced on $\vec t_1$,
for  every structure $s_2$ such that $s_1$ is an induced substructure of $s_2$, there is a tuple $\vec t_2$ 
extending $\vec t_1$ and isomorphism of $s_2$ to $\vec t_2$.
Since  $T_J$ is a complete theory with a computable axiomatization, it is decidable.

The key property of $T_J$ is that for any finite linear order  $L$ over
a domain $D_L$ and
any $a$ in a model $M$ of $T_J$, $\delta(L)$ can be embedded in $R_a$ if and only
if $|D_L| \not \in J$.

\begin{proposition}
For $k \geq 2$, if $T_J$ is  $k$-RQC
then $J$ is  a pure $k$ order
spectrum.
\end{proposition}

\begin{proof}
Consider the  sentence  $\psi(P)$ in $L_P$ 
stating that:

For some $a$, and some $b$  $Z_{b}$ induces a  linear order  $<_P$ on $P$,
and  $\delta$ applied to $<_P$ is an induced substructure of $R_a$.

The properties of the structure guarantee that such a sentence is true if and only if $|P| \not \in J$.

Suppose $T_J$ is $k$-RQC. The sentence $\psi$ is cardinality-invariant.
Consider embedded finite subsets living within  a totally indiscernible 
set $I$.
Then  there is a $k$-RQ sentence that defines $\psi$, and indiscernibility allows
the  references to the ambient structure to be eliminated from $\psi$.
From this we see  that $J$ is the spectrum of a $k^{th}$ order logic sentence.
\end{proof}

\begin{proposition} For $k \geq 2$ if $J$ is a pure $k^{th}$ order logic  spectrum, then
$T_J$ iss $k$-RQC.
\end{proposition}
\begin{proof} 
Suppose $J$ is the spectrum of $\tau$.

For $G \subseteq P^2$, we write $R_{x,1,P}=G$ to abbreviate the formula:
$\forall y,z \in P ~ R(x,y,z) \leftrightarrow G(y,z)$.
We let $R_{x,2,P}=G$ abbreviate
$\forall y,z \in P ~ R(y,x,z) \leftrightarrow G(y,z)$.
and define $R_{x,3,P}=G$, $Z_{x,i,P}=G$ for $i=1,2,3$  analogously.

Inductively, it suffices to show that  a formula $\gamma(x, \vec y)$ of the form:
$\exists x ~ \rho(x, \vec y)$
with $\rho$ $k$-RQ, is equivalent to a $k$-RQ formula.

We can assume that $x$ occurs in $\rho$  only in $L$ atoms. Since $T_J$ enforces
that $R(x,y,z)$ implies $x,y,z$ distinct, we can also assume $x$ occurs at most once in each $R$
atom.
We then rewrite $\rho$ as
\begin{align*}
\exists ~ G_1, G_2, G_3 Z_1, Z_2, Z_3 \subset P^2 ~ R_{x,1,P}=G_1 \wedge \\
R_{x,2,P}=G_2 \wedge R_{x,3,P}=G_3 \wedge \ldots  \wedge \rho'
\end{align*}
Here $\ldots$ indicates additional equalities with $Z_{x,i,P}$.
And $\rho'$ is obtained from $\rho$ by replacing each $R$ atom containing $x$
with the appropriate $G_i$ atom and similarly for $Z$.
It will suffice to simplify  the formula
\[
\exists x ~ R_{x,1,P}=G_1 \wedge R_{x,2,P}=G_2 \wedge R_{x,3,P}=G_3 \wedge \ldots
\]
where $G_1, G_2, G_3$ are constrained to be in $P^2$.

We claim that for $G_1, G_2, G_3 \subseteq P^2$ and $P$ finite, there is $x$
such that  $R_{x,1,P}=G_1$, $R_{x,2,P}=G_2$ and $R_{x,3,P}=G_3$ and the $Z_{x,i,P}$
equalities hold exactly when
no $B(n)$ for  $n \in J$ embeds as an induced substructure of $G_1$.
This is easily verified: the fact that $G_2$ and $G_3$ play no role 
uses the assumption that no $B(n)$ is contained in a cone. The equalities
involving $Z$ can always be achieved.

Note that if $G_1$ contains some $B(n)$ then necessarily $n \leq |P|$,
since $G_1$ is contained in $P^2$.
Thus the condition that $G_1$ does not contain any $B(n)$ can be expressed as:

for no subset $P'$ of $P$, for  no linear order $<$ on $P'$
if $P'$ satisfies $\tau$, then $\delta(<)$ does not embed as an induced substructure of $G_1$.

Since $\tau$ is a $k$-RQ formula in $P$, 
this is a $k$-RQ formula on $G_1, P$.

\end{proof}

Thus we have reduced  separating levels of RQC to the corresponding
separations in finite model theory.
By Proposition \ref{prop:spectrum} we can
choose $J$ to be the spectrum of a $k+2$ sentence but not a 
$k$ order logic sentence, and thus get
theories that are $k+2$-RQC but not $k$-RQC.
By adjusting
the definability of the set of natural numbers $J$, we can also show:

\begin{corollary} There are decidable theories that are not  $\omega$-RQC.
\end{corollary}

\begin{proof} Choose $J$ that is decidable, but  not a pure $k$ spectrum for any $k$.
\end{proof}

By  varying the complexity of the set $J$ (while retaining decidability) we can show that the complexity of evaluating formulas can be arbitrarily high.

Recall from Section \ref{sec:def} that there is a strong connection between various classical model-theoretic  dividing
lines and $1$-RQC: for example,  $1$-RQC implies NIP.
The result above can be seen as a negative one about similar connections between higher-order collapse
and classical model-theoretic dividing lines:
Note that model-theoretically the different theories $T_J$ are quite similar.
But they can be adjusted to get any level of RQC or none at all, depending on $J$.

The construction does not give theories that are NIP. We can show that it is possible
to separate $2$-RQC from $1$-RQC with an NIP theory: see the appendix.  But the problem of separating
  $k+2$ from  $k$ with an NIP theory for $k>0$
 is open.

\myeat{
We have also shown:

\begin{corollary} There are theories that are finitely redeemable
but not $\omega$-RQC.
\end{corollary}
This is because, as noted previously, isomorphism-invariant formulas
cannot be NFR.
}

%% file: highcomp.tex
\subsection{Conversion of isomorphism-invariant sentences to higher-order restricted-quantifier form and the complexity of theories} \label{subsec:complexitytheory}
The previous family of examples was constructed so that there is a tight
connection between the level of RQC  of general $L_V$
formulas and the descriptive
complexity of the underlying set of integers $J$, which also controls 
the complexity of the theory.
Is there  a more general connection between higher-order RQC for  a theory $T$ and the complexity
of decision procedures for $T$?

For an arbitrary theory there may be no connection between the complexity of
the theory and the ability to convert $L_V$ sentences to RQ ones: this is the
case, for example in Example \ref{ex:fcp}.
But the following result shows that if we focus on isomorphism-invariant sentences, then
collapse can follow  from reasonably weak assumptions on the
complexity of a decision procedure for the theory.
We tailor the result for isomorphism-invariant $\omega$-RQC, but variations hold for
$k$-RQC.
\begin{proposition}\label{prop:elem}
Consider a complete theory $T$ in a language $L$ and let $C$ be an infinite set of 
constant symbols disjoint from $L$.
Assume that valid sentences of $L \cup C$ can be  efficiently enumerated and
also that for any number $v$, there is an elementary
time algorithm that decides whether a first-order sentence $\phi$  in $L \cup C$
with  at most $v$ bound variables is consistent with $T$.
Then every isomorphism-invariant sentence in $L_V$ 
is equivalent to a $k$-RQ
one for some $k$.
\end{proposition}

\begin{proof} Consider the following algorithm for determining if $\phi$ holds on an embedded finite
model: given a finite interpretation $I$ of relations in $V$ whose active domain has $n$ elements, we form an object coding  the $L \cup C$ sentence
\[
\gamma_I=\bigwedge_{i,j \leq n}  c_i \neq c_j \wedge \phi_I
\]
where $\phi_I$ is formed from $\phi$ by replacing every atom $R(\vec x)$ by $\bigvee_{\vec c \in R^I}
\bigwedge_k x_k = c_k$.
$\gamma_I$ is of size polynomial in $I$, and $\phi_I$ can thus be coded by suitable sequences
over the domain of $I$.
By assumption, there is an elementary time algorithm deciding whether a 
$\phi_I$ is consistent with $T$. The run of such an algorithm on $\phi_I$ can be captured again by a suitable
higher order object.  Since $\phi$ is isomorphism-invariant, being satisfied on some $n$ element subset implies
being true on $P$.
\end{proof}

However, note that the model-theoretic sufficient conditions
given in Section \ref{sec:def} imply that RQC can hold for theories that are not even decidable.

%% file: robust.tex
\section{Complexity disconnections and weakening RQC by allowing Expansion of the signature} \label{sec:robust}
In Section \ref{sec:higher} we dealt with examples where there was some correlation between the expressiveness of formulas
and complexity of evaluation.   We provided examples which not only fail to have $k$-RQC, but examples where logical formulas are highly expressive, and thus complex to evaluate. 

On the other hand, the failure of RQC or $k$-RQC alone does not mean that evaluation of formulas is complex. It might be that a structure fails to have RQC not because it is too expressive, but because it is not expressive enough.
Like quantifier-elimination, RQC is sensitive to the signature.
It is easy to construct theories that are ``badly behaved'' -- not even $\omega$-RQC --
but which become $1$-RQC  when the theory is expanded.

Before we give an example of this, we give a simple tool for verifying failure of $\omega$-RQC:

\begin{definition} \label{def:nfr}
We say that an $L_V$ formula $\phi$  is
``not finitely redeemable'' (NFR) if:

for each $k$, there are $P_k,P'_k$ and a function $f_k$ taking $P_k$ to $P'_k$
such that: $(M,P_k) \models \phi$, $(M,P'_k) \models \neg \phi$,
but  $f_k$ preserves all formulas with at most $k$ variables. 
\end{definition}

NFR formalizes the idea that there are embedded subsets distinguished by $\phi$ that are increasingly similar locally.
It is easy to see:

\begin{proposition} \label{prop:nfr} An NFR formula
cannot be converted to any higher-order restricted-quantifier sentence. Thus if $T$ has an NFR formula it does not have $\omega$-RQC.
\end{proposition}
\begin{proof} Fix any higher order RQ sentence $\phi'$. Let $k$ be such that every $L$-subformula has at most $k$ variables.
Then $\phi'$ would be preserved by $f_k$ above.
\end{proof}

We now turn to the example.

\begin{example} \label{ex:transient}
Let us return to Example \ref{ex:fcp},
an equivalence relation with classes of each finite size. 
We first argue that the $L_P$ formula $\phi_{\contained}$  is 
not equivalent to any $k$-RQ formula.
Fix the standard model of this theory, call it $M$.

We proceed by showing that $\phi_{\contained}$ is NFR, hence not $\omega$-RQC.
We choose $P_k$ to be a set of $k$ elements that exhausts an equivalence class of size $k$, and $P'_k$
to be a set of $k$ elements inside an equivalence class of size larger than $k$. We can then choose $f_k$
to be an arbitrary bijection between $P_k$ and $P'_k$.

Let $T^+$ expand $T$ to $L^+=\{E, <(x,y)\}$, stating that $<$ is a linear order
for which each $E$ equivalence class is an interval with endpoints. Then $T^+$ can
be shown to be $1$-RQC. For example, the sentence $\phi_\contained$  now becomes equivalent to
a restricted-quantifier one: intuitively, we say that $P$ contains the interval between two equivalent elements $a$ and $b$ which form
the endpoints of a class.
The full argument is included in the appendix.
\end{example}

We thus investigate define another weakening of RQC, allowing expansion of the interpreted signature. We wish to see if adding this ability to expand gives us a closer correlation between the collapse notion and complexity.
\begin{definition}
Say that an $L_V$  sentence $\phi$ is \emph{potentially $k$-RQC} (for $T$) if there
is some expansion of $T$ where $\phi$ is equivalent to a $k$-RQ sentence.
We say that a theory is \emph{potentially $k$-RQC} if there is an expansion of $T$ that
is $k$-RQC, and define potentially  $\omega$-RQC analogously.
A theory that is not potentially $\omega$-RQC  is said to be \emph{persistently
unrestricted}. 
\end{definition}
Thus the persistently unrestricted theories are ones in which this weakening does not help us obtain RQC.
Notice that the ``potentially RQ'' sentence $\phi_{\contained}$ in Example \ref{ex:fcp} is not isomorphism-invariant.
This is not a coincidence:
\begin{restatable}{proposition}{transient} \label{prop:transient}
If theory $T$ has an isomorphism-invariant $L_V$  sentence $\phi$  not
equivalent to a $k$-RQ  sentence, then the same is true in any consistent expansion
of $T$.
\end{restatable}

The argument is given in the appendix.

The proposition tells us that weakening RQC in this manner can not allow us to use unrestricted quantification  to do  any new ``pure relational computation''.
In particular, really badly-behaved theories like full arithmetic cannot become $k$-RQC via augmenting the signature.

An open question is:
\emph{Is every NIP theory potentially $\omega$-RQC?}
One cannot use isomorphism-invariant sentences to get a counterexample, since by Theorem \ref{thm:bbiso}
from \cite{baldben}, isomorphism-invariant formulas in NIP theories are always equivalent
to $1$-RQ ones that use an additional order, and thus in particular are equivalent to $2$-RQ ones.

Although we do not have an NIP theory that is persistently unrestricted, the main result of this section is that there
are  persistently unrestricted theories that are reasonably well-behaved.
We show that even when isomorphism-invariant formulas are well-behaved, and 
the complexity of deciding the theory is elementary (see \cite{kozenbooleanalg}), 
the theory may be persistently unrestricted, Thus even with this weakening, failure of the (weakened) collapse does not imply high complexity of evaluation:

\begin{theorem} \label{thm:abapersistentlyunbounded}
In the theory of atomless Boolean Algebras, each $L_V$ formulas that is isomorphiism-invariant are equivalent to ``
an $2$-RQ formula. 

In fact,  there is an infinite
set such that every formula is equivalent to a $2$-RQ one for embedded finite models coming from this set. Thus
every isomorphism-invariant $L_V$ sentence is equivalent to a $2$-RQ one.
In addition, isomorphism-invariant
$L_P$ sentences are equivalent to $1$-RQ ones. But the theory  is persistently unrestricted.
\end{theorem}

We will first sketch an argument that there is an infinite  set $A$ such that, for embedded finite structures with
active domain of the $V$ relations in $A$, every $L_V$ sentence is equivalent to a
Monadic Second Order Logic RQ-formula. In fact, we will describe such a set where the elements are totally indiscernible: the $L$ formulas $\phi(\vec a)$ satisfied
by tuples $\vec a$ of distinct elements are independent of the choice of $\vec a$.
Hence for $P$ unary, isomorphism-invariant $L_P$ sentences are equivalent to first-order sentences quantifying
over $P$.

Fix an $L_V$  $\phi$ that depends only on the isomorphism type of the $V$ structure. Let $\mathcal C$ be the class of embedded finite models
where the active domain is a subalgebra of the underlying substructure.
We observe:

\begin{claim} \label{clm:fv}
Every $L_V$ $\phi$ is equivalent, for embedded finite
models in $\mathcal C$, to a $1$-RQ formula. 
\end{claim}
\begin{proof}
Consider a finite subalgebra $P_0$. Note that the full Boolean Algebra factors as $P_0 \times B$, where
$B$ is just a copy of the full algebra. Thus by the Fefferman-Vaught theorem for products
\cite{hodgesbook}, a first-order sentence over $P$ can be decomposed into a Boolean combination of sentences
over $P_0$ and over $B$, uniformly in $P_0$. The sentences over $P_0$ are just $1$-RQ sentences,
while  the sentences over $B$ are independent of $P_0$, hence they
can be replaced with true or false.
\end{proof}
An alternative inductive transformation proving the claim is given in the  appendix.


Now take $A$ to be an infinite antichain  in a model. We argue that $A$ has the required property.
The antichain is definable within the subalgebra that it generates by a $1$-RQ formula: it is just  the atoms
of that algebra. We transform $\phi$ to $\phi'$
as follows.
We add to the $V$ signature a unary predicate $P$, and add a conjunct $\phi_{\kw{sub}}$ saying that
the domain of $P$ is a subalgebra, and that its atoms are exactly the active domain of $V$.
We replace references to a $V$ relation symbol $U$ in $\phi$ with references to the 
restriction of $U$ to elements that are atoms of the active domain
of $P$. For example, if $\phi = \exists x y ~ \neg U(x,y)$, then $\phi'$ would be
$\phi_{\kw{sub}} \wedge  \exists x y ~ \neg  (\mbox{ $x$ $y$ are atoms of the  domain of $P$ that satisfy $U$ })$.
It is easy to verify that  $\phi'$ holds on an interpretation of relational
vocabulary $U \cup \{P\}$ 
where $P$ is the subalgebra generated by the active domain of $U$,
exactly when $\phi$ is true on the original interpretations.
$\phi'$ can be converted to a $1$-RQ formula $\phi''$ over $\mathcal C$.
Then quantifications in $\phi''$, where quantifications  range  over the subalgebra generated
by the active domain, can be replaced by references to subsets of the atoms, giving
a Monadic Second Order restricted-quantifier formula.

Summarizing, we have shown each claim in Theorem \ref{thm:abapersistentlyunbounded} except the last one.
We now sketch the argument that the theory is persistently unrestricted.

We start by considering one particular extension $\alba^<$, which we describe
as the complete theory of a particular model. In this model  we have a distinguished
partition of $1$ into countably many elements ordered by a discrete linear order $<$.
 We look at the Boolean Algebra generated by these sets, which we can identify
as infinite bit strings over $\nats$; we order two elements via the lexicographic
ordering on   the corresponding strings.
 We next consider the $L_P$ formula $\phi_{\isunion}$ stating that $P$ has an element $x_P$ that is the union
of all other elements of $P$. 
We claim  that $\phi_{\isunion}$ is NFR within $\alba^<$. Let us recall the definition of an $L_V$ formula
being NFR, Definition \ref{def:nfr}.
This means that there is a family of pairs of embedded finite subsets $P_n$ and  $P'_n$, 
along with  a mapping $f_n$ from $P_n$ to  $P'_n$
such that $f_n$ preserves all $L$-formulas with at most $n$ free variables, but each $(M,P_n)$ disagrees with $(M,P'_n)$ on $\phi$.
It is easy to see that $\alba^<$ admits such a family: we let $P_n$ be a set of $n$ disjoint Boolean
Algebra elements in the linear ordered set, along with an additional element
representing their union.  While $P'_n$ consists of all but  the $<$-last of the $n$ elements, along
with the union of all the $n$ elements.  Here ``all $L$ formulas'' can
be replaced by ``all atomic $L$ formulas'', using quantifier elimination in $\alba$.
We  note that NFR wtnesses
the failure of $\omega$-RQC, as mentioned in  Proposition \ref{prop:nfr}.

We next make use of a result from \cite{definabilitypatterns}.
For a tuple $\vec x_0$ in a $L'$ structure $M'$ and $L$ a subset of $L$,
the $L$-type of $\vec x_0$ is the set of formulas $\phi(\vec x)$ with vocabulary in $L$ that
are satisfied by $\vec x_0$.
A theory $T$ is \emph{everywhere Ramsey theory} if
for every  model $M$ of $T$, for every expansion $M_e$ of $M$,
there is an elementary extension $M^*_e$ of $M_e$ and a copy
$M'$ of $M$ such that for every $k$, the $L(M)$-type of a $k$-tuple in $M'$ determines
its type in $M^*_e$. 
The result we use is  Example 5.7 from \cite{definabilitypatterns}s:
\begin{proposition}
$\alba^<$ is an everywhere Ramsey theory.
\end{proposition}
Further discussion on Ramsey properties of $\alba^<$, and the corresponding class of finite lex-ordered subalgebras, can be found in \cite{kpt}.

Now consider an arbitrary expansion $M^*$ of a model of $\alba$.
We can further expand $M^*$ to a model of $\alba^<$ by choosing the ordering
appropriately. Let $M, P_n, P'_n, f_n$  be the witness finite models
for  $\alba^<$ above.
Since $\alba^<$ is everywhere Ramsey, 
we  find a copy of $M$ inside an elementary extension of $M^*$. The copies
of $P_n, P'_n, f_n$ will witness that $M^*$ is also not  $\omega$-RQC.


%% file: pseudofinite.tex
\section{Failure of weaker forms of collapse for finite fields and connections to complexity of decision procedures} \label{sec:pseudo}

Section \ref{subsec:separate}
provided us with examples of theories 
that have different levels of higher-order collapse, and also varying complexity of evaluation. It is natural to ask about decidable theories
arising naturally, both in terms of collapse and complexity. In  Section \ref{sec:robust} we saw examples of decidable structures that behave badly with respect to collapse -- not even $\omega$ RQC --  but this did not reflect something fundamental about complexity. This was true in the case
of Example \ref{ex:transient}, since collapse could be fixed by expansion.

What about other well-studied decidable theories.
In previous work it was noted that many decidable examples (e.g. real-closed fields) 
satisfy broad model-theoretic properties like o-minimality,
known to imply $1$-RQC. 
 We have seen that random graph is $2$-RQC.
Another well-known  source of decidable theory is the theory of \emph{finite fields}: the set of sentences that hold in all finite structures. This was shown decidable by Ax \cite{ax}.
This is an incomplete theory, but there are decidable completions, the theories of \emph{pseudo-finite fields}: infinite
fields that satisfy all the sentences
holding in every finite field. There are several alternative definitions.
For example,  they are the ultraproducts of finite fields \cite{ax,fieldarith}: a pseudo-finite
field of characteristic $p$ is an ultraproduct of finite fields of characteristic $p$, while a pseudo-field of characteristic
$0$ is an ultraproduct of finite fields with unbounded characteristic.
Pseudo-finiteness can be axiomatized, giving another decidable incomplete theory \cite{ax}: the assertions we make below will apply to all pseudo-finite
fields  -- equivalently, all completions of the axiomatization.

Some of our results will focus on
 positive characteristic. We let $F_p$ be the theory of finite fields in characteristic $p$ and  $PF_p$ be the  theory of pseudo-finite fields in characteristic
$p$.  
 We begin by showing that these
   are not $k$-RQC for any $k$, even  for isomorphism-invariant formulas.
This result  shows that natural decidable theories can fail to have even weak forms of RQC.  It is arguably surprising
given that pseudo-finite fields have some commonalities with
the $2$-RQC random graph. They are both canonical examples
of  ``simple but unstable'' (hence not NIP) theories \cite{wagnersimple}.

In fact,  we will show that in pseudo-finite fields, and in sufficiently large finite fields, we
can interpret the $n^{th}$ iterated powerset over
$P$, for all $n$, uniformly in the field.
From there it is easy to show, since we can code $k$-order logic for any $k$,
that such fields cannot be $k$-RQC for any $k$.
We will use this result to derive new lower bounds on decision procedures for pseudo-finite fields and for finite fields.

 By a $k^{th}$-order finite set theory structure over a predicate $P$, we mean a structure 
 with language $(P_1,\ldots,P_k,E_1,\ldots,E_k)$, with $P_1=P$, $P_i$ a unary predicate,
 $E_i \subset P_i \times P_{i+1}$ a binary predicate for a membership relation, such that
 $E_i$ satisfies extensionality and $P_{i+1}$ can be viewed as a set of subsets of $P_i$ closed under flipping membership of one element, i.e. such that:
\begin{align*}
\forall x \in P_i ~ \forall y \in P_{i+1} ~ \exists y' \in P_{i+1}\\
(x E_i y' \leftrightarrow \neg x E_i y)
 \wedge  [\forall z \in P_i ~ z \neq x \rightarrow (z E_i y' \leftrightarrow   x E_i y) ]
\end{align*}  In case the interpretation of $P$
 is finite, this implies that $(P_i,P_{i+1},E_i)$ is  isomorphic to the membership relation on $P_i$ and the set of subsets of $P_i$. Interpreting such a structure is equivalent to interpreting
the $k^{th}$ iterated powerset of $P$.
      
 \begin{theorem}  \label{thm:pseudofiniteposexhaustive} For any $k$, there  exist 
$L_P$ formulas 
\[
\phi_{P_1},\cdots,\phi_{P_k}, \phi_{E_1},\ldots,\phi_{E_k}
\]
with the following properties.
\begin{compactitem}
\item each $\phi_{P_i}$ has one distinguished
free variable, and each $\phi_{E_i}$ has two distinguished free variables.
Each formula may  optionally have
additional ``parameter variables''.
\item  for any $M \models PF_p$,  
and any interpretation of $P$ as a finite subset of $M$, 
for some elements $(c_1,\ldots,c_j)$ of $M$  for the parameter variables
 the $\phi_{P_i},\phi_{E_i}$  define the full $k^{th}$ order set theory structure over $P$. 
\end{compactitem}
Further, there is an  $L_P$  formula $good$ such that the parameters can
be taken to satisfy $good$, and every parameters satisfying $good$ have
the properties above.
\end{theorem}

\begin{proof}
  Consider the additive group $A$ as an $\Ff_p$-vector space. We make use of the
fact that there is a definable
  non-degenerate bilinear  form $\bilin$ into $\Ff_p$. For example,
one can use the trace product (see Lemma 3.9 of \cite{dugaldcharlie}).    Now 
consider $X_0$ a finite subset of $A$, of size $n_0$.    Define $X_i,Y_i \subset A$ inductively.  
We will use  as parameters elements $t_1,t_2,\cdots$, taken to be 
   algebraically independent over the field generated by $X_0$.  We will always have
  $Y_i$ be the subspace generated by $X_i$:
 \[ 
y \in Y_i \iff (\forall x)[(\forall u \in X_i ~ \bilin(x,u) = 0) \to \bilin(x,y)=0] 
\]

Inductively we set
 \[X_{i+1} = \{ \frac{1}{t_i-y}: y \in Y_i \} \]

Then $X_{i+1}$ is linearly independent, and we can generate $Y_{i+1}$ as
above.
Note that for a non-degenerate biliinear form, the linear span of a finite
set $P$ is just the orthogonal complement of the orthogonal complement: the
set of elements $x$ such that $\bilin(x,a)=0$ whenever $\bilin(a,p)=0$ for each $p \in P$.
Thus there is an $L_P$ open formula that defines the linear span of a finite
set $P$.
Given a finite subset $J$ of $X_{i+1}$
with its sum $e_J$, we can recover $J$ from $e_J$ as the set of $j \in J$
such that $e_J$ is not in the span of $J-\{j\}$: this definition is correct
because $J$ was linearly independent.
Composing with the inverse of   the map $ \frac{1}{t_i-y}$, we
can see that $X_{i+1}$ codes subsets of $X_i$.
 Thus $i^{th}$ order monadic logic over $X_0$  is interpretable in $X_i$, and
by consider $k'$ sufficiently higher than $k$ we can interpret full $k^{th}$ order logic.

We now discuss how to capture the parameters by a formula $good$ as
required in the last part of the theorem.
Let $T_i(P)$ be the theory of $i^{th}$ order set theory over a  predicate
$P$ interpreted by a finite set. $T_i(P)$ is finitely axiomatizable. For example for $i=1$ we have a predicate $\epsilon(x,y)$
and we need to say that there is $y$ with no elements, and for each $y$ and each $p \in P$
there is a $y'$ such that $y'$ agrees with $y$ on $P \setminus \{p\}$ and disagrees with
$y$ on $p$.
Let $good_i(\vec t)$ be an $L_P$ formula
expressing that the  formulas $\phi_1(\vec t) \ldots $ with
parameters $\vec t$ defined above (for $i$) give an interpretation
satisfying $T_i(P)$. Since $T_i(P)$ is finitely axiomatizable, 
this is a first-order sentence in $L_P$. The result above shows
that  the theory of pseudo-finite fields entails $\exists \vec t ~ good_i(\vec t)$, for each $i$.
\end{proof}
Note that above we talked about an interpretation with equality interpreted
in the standard way, but where the formulas in the interpretation can include parameters. The standard notion of interpretation \cite{hodgesbook}
allows the output model to be defined as a quotient by a definable  equivalence 
relation. 
That is, the interpretation can interpret equality. We can remove the notion
of ``good parameter'' by introducing such an equivalence relation, thus
talking about an interpretation without parameters.
In either case, the consequence is:

\begin{corollary} \label{cor:expressho}
Any $k$-RQC  sentence $\gamma$ over $P$ can be expressed
using an $L_P$ sentence:  by
existentially quantifying over $\vec t$ satisfying $good_k$ and checking $\gamma$ in the resulting interpretation.
\end{corollary}

\myeat{
Note that the existence of a non-degenerate bilinear form
and the construction above 
already
suffices to prove that pseudo-finite fields of positive
characteristics have IP, and hence are unstable, a result
established via alternative means by 
Duret \cite{duret}.
}

 
\begin{corollary}  For every number $k$, any completion of the
theory of pseudo-finite fields in positive characteristic cannot be $k$-RQC (even for
isomorphic-invariant sentences).
\end{corollary}
\begin{proof}  
By   Proposition \ref{prop:spectrum} for each $k$ there is a set of integers 
that is the spectrum of  a $k+2$ order logic sentence $\phi_k$ over the empty vocabulary,
which is not the spectrum of a $k$ order logic sentence over the empty vocabulary.
By the result above, 
  $\phi_k$ can be described in $L_P$. But clearly $\phi_k$ it cannot
be described by a $k$-RQ sentence of $L_P$, since over an indiscernible 
set the $L$ atoms could be eliminated.
\end{proof}


\subsection{Embedded finite model theory and lower bounds}
Recall that Proposition \ref{prop:elem} indicated that
if $\omega$-RQC fails for an isomorphism-invariant sentence, then the complexity of deciding the theory must be very high. 
We show that efficiently interpreting iterated powersets can also suffice to give lower bounds on the complexity of the theory, illustrating this
in the case of pseudo-finite fields, and also for the incomplete theory of fnite fields. 
We doubt that one can get an isomorphism-invariant counterexample to $\omega$-RQC
for pseudo-finite field: see discussion in the next subsection.

A primitive recursive decision procedure for the theory
is obtained in Fried-Sacerdote
\cite{FriedSacerdote}.
 But to our knowledge it has not been improved by concrete bounds.   This contrasts with theories such as real closed fields,
 where doubly exponential bounds are known \cite{bpr}.

From our construction we immediately obtain the following  lower bound:
\begin{proposition} There is no decision procedure for $PF_p$ in positive characteristic
$p$ that works in elementary complexity.
\end{proposition}
\begin{proof}
We have shown that the theory can interpret the $k^{th}$ iterated powerset for any $k$.
This
allows us to polynomially many-one reduce the model-checking problem for 
$k^{th}$ order logic
to satisfiability of sentences in the theory, uniformly in $k$.
Since the model checking
problem for $k^{th}$ order logic is known to be hard for an $O(k)$ exponential
tower in $k$ (see, e.g.  \cite{hellahigherorder}) we can conclude.
\end{proof}

We now draw the promised conclusion for the \emph{theory of finite fields}. Recall that this is the incomplete theory consisting of all sentences that hold in
every finite field.
\begin{proposition} There is no decision procedure for the theory of finite fields in positive characteristic
$p$ that works in elementary complexity, and thus there can be no such procedure for theory of finite fields (without restricting the characterisitc).
\end{proposition}

\begin{proof} Fix our interpretation $\iota_k$ of the $k^{th}$ order powerset of $P$ from above. Not every
finite field of characteristic $p$ will interpret $\iota_k(P_0)$ as the powerset, but by the compactness theorem, cofinitely many of them will.
Furthermore we can efficiently write a sentence $\kw{verify}_k$ that checks whether a hierarchical structure is indeed the $k$-iterated powerset of $P_0$: $verify_k$
checks that the empty set is represented, and that if a set is represented, so is any set formed by changing one membership.
We are now ready to reduce model checking a $k$-higher-order logic sentence $\psi$ to entailment in finite fields: we generate
the sentence $\kw{verify}_k(\iota_k(P_0) \rightarrow \psi(P)$.
Roughly, `` if the $k^{th}$ power set  exists then $\psi$ holds in it''.
This sentence will be true in all finite fields iff $\psi$ actually holds in the $k^{th}$ iterated power set of $P$.
\end{proof}

\subsection{Lower bounds for quantifier-elimination}
We can also make a conclusion about quantifier elimination.
We will state the results for pseudo-finite fields, but they could be restated in terms of finite fields.

A \emph{parameterized algebraically-constrained formula} is a formula 
\[
\exists y_1 \ldots y_k ~ \bigwedge_i F_i(\vec x, y_1 \ldots y_k,\vec p)=0
\]
where 
$F_i$ are sums of monomial terms in $\vec p, \vec y, \vec x$, such that, in any finite field, for each $\vec x$ and $\vec y$, the number of witnesses $\vec y$ that satisfy the formula is at most polynomial in $k$ and the maximum degree of the $F_i$.

\begin{restatable}{theorem}{pseudofiniteqe} \label{thm:pseudofiniteqe} For every formula $\phi(\vec x)$ in the language of  fields there is a disjunction of paramaterized algebraically-constrained formulas $\phi'(\vec x, \vec p)$ and a formula $\theta(\vec p)$ such that
in every pseudo-finite field $M$ there is some $\vec p_0$ satisfying $\theta$, and for any such $\vec p$
$\phi(\vec x)$ is equivalent to $\phi'(\vec x, \vec p)$ in $M$.
\end{restatable}

We provide a proof in the appendix, but we stress that similar quantifier elimination results
are known from the literature. See, for example,
Section 2 of \cite{zoelouangus} and \cite{ax}.

 \begin{theorem} \label{thm:pseudofinposlowerbound} The complexity of quantifier-elimination turning
a  $PF_p$ formula into a positive Boolean combination of basic formulas cannot be bounded by a stack
 of exponentials of bounded height; this is already true for $PF_p$ for any fixed prime $p$.
 In fact for any $k$ there exist  formulas $\phi_n$ with a bounded number of  variables -- and
hence with a bounded quantifier rank, independent of $n$ --  of length $O(n)$,
 such that any basic formula $\psi_n$  which is $PF_p$ equivalent to $\phi_n$, must 
 have size at least  $p^{p^{\cdot^{\cdot ^{\cdot ^{p^n}}}}}$ (exponential tower of height $O(k)$).  \end{theorem}

Before giving the proof of the proposition, we present the main tool that we use:

We will argue the following:
\begin{theorem} \label{thm:bezout} There is an elementary function $F$ such that in any pseudo-finite field
if $\theta(\vec x, \vec p)$ is an algebraically-constrained formula
and for a given $\vec p$ the set $\theta_{\vec p}$ of $\vec x$ such that $\theta(\vec x, \vec p)$ holds is finite,
then  $\theta_{\vec p}$ has size at most $F(size(\theta))$. Here size of $\theta$ takes into account
the degrees of the polynomials
in addition to the number of terms.
\end{theorem}

Theorem \ref{thm:bezout} can be thought of as an analog of Bezout's theorem bounding the number
of joint solutions of polynomial equalities. The subtlety is  that while Bezout's theorem works
over algebraically closed fields, here we are working in pseudo-finite fields, which will never be  algebraically
closed. The above result
follows from
 \cite{emm}, Prop. 2.2.1. We now give the details.

Recall that an algebraically-constrained formula is an existential quantification of a conjunction of polynomial equalities,
The result can be restated in terms of \emph{finite fields}:
if we we have a basic formula $\theta$ where the size of solutions is bounded as we range over all finite fields 
and parameters $\vec p$, then the size must be bounded by an elementary function.
We first argue for a similar statement that removes the existential quantification:

\begin{theorem} \label{thm:bezoutatomic}
There is an elementary function $F$ such that in any pseudo-finite field 
if $\theta(\vec x, \vec p)$ is a conjunction of polynomial equalities
and for a given $\vec p$ the set $\theta_{\vec p}$ of $\vec x$ such that $\theta(\vec x, \vec p)$ holds is finite,
then $\theta_{\vec p}$ has size $F(size(\theta))$. As in Theorem \ref{thm:bezout}, size takes account the degrees of the polynomials
in addition to the number of terms.
\end{theorem}

We need some auxiliary definitions. 
The theory of  Algebraically Closed Fields with an Automorphism 
(ACFA) is in the language of an algebraically closed field $F$ with a distinguished automorphism $\sigma$ satisfying
certain axioms, one of which is that it is a \emph{difference-closed difference field}.

The axioms will not concern us. The only thing that is important to us is:

\begin{proposition} \label{prop:acfa}
Every pseudo-finite field is the fixed field of $(K, \sigma)$ that is a model of ACFA.
\end{proposition}

See, for example, \cite{garcia}, comments after Proposition 4.3.6.
The proposition above allows us to reduce reasoning about polynomial equalities within a pseudo-finite
field, which is not algebraically closed, to reasoning within the fixed field of an algebraically
closed field with a suitable automorphism. We can thus think of our solution set as the conjunction
of the polynomial equalities along with the ``difference equation'' $\sigma(x)=x$.

An \emph{affine variety} over a field $K$ is a set of elements $K$ defined by 
polynomial equations with coefficients in the field.  So we are interested in the intersection
of an affine variety with a simple difference equation.

We now state Proposition 2.2.1 of \cite{emm}:

\begin{proposition} \label{prop:emm} Let $(K;  \sigma)$ be 
a difference-closed difference field; and let $S$ be a
 subvariety of $\mathcal{P}^l$
 defined over $K$. Let $Z$ be the Zariski closure of 

\[
\{x \in \mathcal{P}:(x,  \sigma(x) \ldots \sigma^{l-1}(x)) \in S\}
\]
Then $deg(Z) \leq deg(S)^{dim(S)}$. 
In particular; $Z$ has at most $deg(S)^{dim(S)}$
irreducible components.
\end{proposition}

We explain the route from Proposition \ref{prop:emm} to Theorem \ref{thm:bezoutatomic}, first introducing
what we need to know about the relevant terminology from Proposition \ref{prop:emm}.
Proposition \ref{prop:emm} is stated for a subvariety of a projective variety: this is
$\mathcal{P}^l$. Again, we will not
need the definition of projective varieties, but note that an  affine variety 
can be embedded
in a projective variety in a way that preserves number of solutions (by adding additional variables and
equations). Thus the same proposition also holds for an affine variety $S$ sitting within
another affine variety.

Proposition \ref{prop:emm} also refers to the Zariski closure, degree, and dimension.  But
if a variety is finite, its Zariski closure is just itself, the   degree is the number of points, and
the dimension is upper bounded by the number of variables in the defining polynomials.

With this in mind, we begin the proof of Theorem \ref{thm:bezoutatomic}:
\begin{proof}
We  apply Proposition \ref{prop:acfa},
to fix an algebraically closed field $K$ and an automorphism  $\sigma$ with $(K, \sigma) \models ACFA$,
 such that our pseudo-finite field is the  fixed field
of $K$ under $\sigma$.
Thus the set we are interested in is 
\[
\{x | \{(x, \sigma(x) \ldots \sigma^{l-1}(x)) \in S\}
\]
 where $S$
is the subvariety of the product of our  variety with itself, adding
the additional equation saying that the components are all equal.
Proposition \ref{prop:emm} tells us that the degree of this subvariety is  bounded by
 $deg(S)^{dim(S)}$. By the comments above about degree and dimension for finite set,
this is an exponential in the size, so we are done.
\end{proof}

We are now ready to begin the proof of Theorem \ref{thm:pseudofinposlowerbound}:
 \begin{proof}    
Fix $k$. 
As in Theorem \ref{thm:pseudofiniteposexhaustive}, 
 we can write a formula $\phi^k(x;P)$ of $L_P$ of size $O(k)$ such that whenever $P$ is 
 interpreted by a set of size $n$  elements, then $\phi^{k}(x;A)$ defines a set whose  size 
  is on the order of
 $n^{n^{\cdot^{\cdot ^{\cdot ^n}}}}$, where the ellipses denote an exponential tower of height 
$k'$.
 For $q$ a power of $p$, let 
   $\phi_{k,q}(x) = \phi_{k}(x; x^{q}=x)$ be the formula obtained from $\phi_k$ by replacing $P$
 by $x^q=x$. Note that $x^q=x$ has precisely $q$ solutions, and so $\phi_{k,q}$ has  on the order of
$q^{q^{\cdot^{\cdot ^{\cdot ^{q}}}}}$
 solutions, where again the ellipses denote an exponential tower of height $k'$.
Note that  these formulas have bounded quantifier rank.
 
Consider a disjunction of algebraically-constrained formulas  $\psi_q$ produced by $\phi_{k,q}$.
It suffices to show: the number of solutions of such a formula is at
most elementary in its size. For this, it suffices to show this for a  single algebraically-constrained formula.

   The projection $(x,y) \mapsto x$
maps  the solutions of the atomic formulas
$\theta(x, y)$ in an algebraically-constrained formula equivalent to  $\psi$
onto those of $\psi$. By the definition of algebraically-constrained, this map is at most $\ell$ to $1$, where $\ell$ is the maximum degree of polynomials in the representation.
So it suffices to bound the number of $x$ in the quantifier-free formulas of the form $\theta$.
We obtain this via Theorem \ref{thm:bezoutatomic}.

\myeat{
the proof of 2.2.1 can also use the maximal degree of these variables.  
Alternatively Bezout's Theorem (see, e.g.,  Lemma 2.1.2 of \cite{emm}) can 
be used to bound the degree of the variety;   this will cost one exponential in the degree, 
leading to $(2^d)^{2^n} = 2^{d 2^n} $.
}

 \end{proof}

        \subsection{$\omega$-RQC and Pseudo-finite fields}
   We continue  with the theory $T_p$ of pseudo-finite fields of  positive characteristic $p>0$.
\myeat{
 The induced language from a field $F$ on a  subset $P$ is generated,
   by definition, by the $n$-ary $0$-definable subsets of $F$, restricted to $P$.  Though 
  the language of rings is finitely generated, the induced language on finite subsets $P$ is not 
  finitely generated uniformly in $P$.     
}
Our next goal is:
\begin{theorem} \label{thm:psfnotomegabounded}
The theory of a pseudo-finite field of positive characteristic $p$ is not $\omega$-RQC.
\end{theorem}
We show that  there is a $L_P$-sentence $\psi$ such that 
for arbitrary large $n$, there are embedded finite subsets $(M,P_n)$, $(M,P'_n)$ with
$M$ a pseudo-finite field and a function $h_n$ from $P_n$ to $P'_n$ that preserves all atomic
formulas among $n$ elements.
Rephrased, we show that we can find {\em isomorphic} 
 pseudo-finite substructures of some   pseudofinite field of characteristic $p$ such that $\psi$ holds in one but not the other.  
 
 We will assume $p>2$; but an analogous construction works if $p=2$, replacing the use of square roots in the argument
below by  Artin-Schreier roots (solutions  $y$ of $y^2+y=x$).

   \begin{lemma}  \label{lem:ram0}  Let $k$ be a field of char. $p\neq 2$.  Let $K=k(x_1,\ldots,x_n)$ be the rational function field over $k$ in $n$ variables.
We can consider each $x_i$ as a member of the field, and
 let $x=\sum_{i=1}^n x_i$.  Let $K^{alg}$ be an algebraic closure of $K$. Within $K^{alg}$, let $K_i$ be the 
 algebraic closure of $k_i:=k(x_1,\ldots,x_{i-1},x_{i+1},\ldots,x_n)$.    Then $x$ is not a square
in $K_1 K_2 \cdots K_n$. 
   \end{lemma}

 \prf  For any irreducible $f \in k[X_1,\ldots,X_n]$, we have a valuation $v_f: K \to \Zz$, defined by:
 \[ v_f(g) = j \iff g=   f^j \cdot h \] 
 where $h$ is a polynomial relatively prime to $f$.  The existence and uniqueness of $j$ follows from unique factorization in $k[X_1,\ldots,X_n]$, and the fact that $f$ is irreducible.   It is clear that $v_f(uv)=v_f(u)+v_f(v)$, $v_f$ is $0$ on $k$,
 $v_f(u+v) \geq \min(v_f(u),v_f(v))$.    
 
 Extend $v_{x_1+\cdots+x_n} $ to a valuation on $K^{alg}$,  denoted by $v$. The key property is
that $v$ is $0$ on $K$ and $1$ on $x$.
 
An automorphism $\si$ of a Galois extension $L$ of $K$ over $K$ is   said to   {\em fix  the residue field}  if whenever $v(u)=0$, we have $v(u-\si(u))>0$.   
I.e. $\si$ induces the identity on the residue field.  Let $Aut(L/K,\res(L))$ denote
the group of automorphisms of $L/K$ fixing the residue field.  The set of elements fixed by every member
of
 $Aut(L/K,\res(L))$ is called the {\em maximal
unramified subextension of $L/K$.}  If it equals $L$, we say $L/K$ is unramified or $L$ is unramified over $K$, and otherwise we say
that $L$  is ramified over $K$.   
\myeat{If $L$ is unramified over $K$, then 
it follows from basic valuation theory that every   (Galois) subextension is unramified.   }

 Now $v$ restricts to the trivial valuation of each field $k_i$: every non-zero element maps to $0$.    
Using basic valuation theory (applying the definition of a valuation to minimal polynomials), $v$ is trivial on each $K_i$.  
Hence
 the only automorphism of $K_i$ over $k_i$ that fixes the residue field is the identity, and thus
 the same statement holds for  automorphisms
 of $K_1\cdots K_n$ fixing the residue field.
  So  $K_1 \cdots K_n$ is unramified over $K$.  To prove that
 $\sqrt{x} \notin K_1 \cdots K_n$, it suffices to show that $K(\sqrt{x})$ is ramified over $K$, which we argue next.
 
 Since $x$ is not a square in $K$ (again using unique factorization), $K(\sqrt{x})$ is a Galois extension of $K$
 of degree $2$, and admits the automorphism $\tau$ fixing $K$ and with $\tau(\sqrt{x})=- \sqrt{x}$.  Since $p \neq 2$,
 $\tau$ is not the identity. We will show that $\tau$ fixes the residue field, thus $\sqrt{x}$ witnesses
that $K(\sqrt{x})$ is ramified over $K$.

 An element $c$ of $K(\sqrt{x})$ has the form $c=a+b \sqrt{x}$ with $a,b \in K$. 
If $b=0$ then $v(c-\tau(c))=v(0)=\infty$ and in particular $v(c-\tau(c)) \neq 0$.
  Thus for the purposes of showing that $\tau$ fixes the residue field, we assume $b \neq 0$.
We have $v(a) \in \Zz$,
 and $v(  b \sqrt{x}) = 1/2 + v(b) \notin \Zz$. So $v(a) \neq v(b \sqrt{x})$.
Thus $v(c) = \min \{v(a), v(b)+1/2\}$.
Towards arguing that $\tau$ fixes the residue field, we are interested in $c$ with  $v(c) =0$.
$v(b)+1/2$ cannot be $0$, so if $v(c)=0$, we know
that  $v(a)=0$ and $v(b) \geq -1/2$. Since $v(b)$ is an integer, we conclude that $v(b) \geq 0$.
We have $v(c-\tau(c))=v(a+b\sqrt{x}- (a-b \sqrt{x}))=v(2b \sqrt{x})\geq v(b)+1/2$. Since $v(b) \geq 0$, we deduce
that $v(c-\tau(c))>0$ as required.

This shows that $\tau$ fixes
the residue field, and $K(\sqrt{x})$ is ramified over $K$, finishing the proof of the lemma.
 \eprf  
   
   \begin{lemma}   \label{lem:differonsquare} Let $F$ be a pseudo-finite field of char. >2, and $a_1,\ldots,a_n \in F$ algebraically independent elements.
   Then there exists a pseudo-finite field $F'$ and $a'_1,\ldots,a'_n \in F'$ such that  
   $(F,a_1,\ldots,a_{i-1},a_{i+1},\ldots,a_n) \equiv (F',a'_1,\ldots,a'_{i-1},a'_{i+1},\ldots,a'_n)$ for each $i$, while also
   $F \models (\exists y)(y^2 =\sum a_i)$ iff $F' \models \neg (\exists y)(y^2 =\sum a_i')$.
   \end{lemma}
   
\begin{proof}
Since $F$ is pseudo-finite, it has one extension of each fixed degree, and thus the Galois
group for any fixed degree is cyclic: see, for example \cite{zoenotes}. We can thus create an automorphism of
each fixed degree extension where the set of elements fixed is exactly
 $F$, and since the full algebraic closure is the union
of these extensions,
$F=Fix(\si)$ for some automorphism $\si$ of the algebraic closure $F^{alg}$.

 Let $k$ be the algebraic closure of the prime field, $K=k(a_1,\ldots,a_n)$, $k_i = k(a_1,\ldots,a_{i-1},a_{i+1},\ldots,a_n)$, $K_i = k_i^{alg}$.  By Lemma \ref{lem:ram0}, there exists an automorphism $\tau$ of $F^{alg}$
 fixing $K_1,\cdots,K_n$ but with $\tau(a) = -a$ where $a=\sum_{i=1}^n a_i$.  Let $\si' = \tau \si$.
Let $K$ be $Fix(\si')$. Thus $K$ is a subfield of $F^{alg}$. Further
$K$ has exactly one extension of every degree.
We now use the following result:

\medskip

For $F$ a pseudo-finite  field,
$K$ a subfield of $F^{alg}$  having at most one extension of each degree, there is a pseudo-finite field $K'$ such
that the algebraic numbers in $K'$ are isomorphic to $K$. By taking an isomorphic copy, 
we get that $K$ is a subfield of $K'$ and $K$ is relatively algebraically closed in $K'$.

\medskip

The claim above is an extension of
Proposition 7 in \cite{ax},  which states that for every subfield $K$ of  the algebraic numbers
such that $K$ has at most one extension of each degree, there is a pseudo-finite field $K'$ such
that the algebraic numbers in $K'$ are isomorphic to $K$. 

 Thus there exists a pseudo-finite field $F'$ containing $Fix(\si')$ and such that
 $Fix(\si')$ is relatively algebraically closed in $F'$.  For each $i$ set $a_i' = a_i$.  Then
for any $x$ we have
 $F' \models   (\exists y)(y^2 =x)$ iff $Fix(\si') \models   (\exists y)(y^2 =x)$.
In the case $x=a$, we have this is true
 iff $\tau \si (\sqrt{a} ) = \sqrt{a}$
 iff $- \si (\sqrt{a}) = \sqrt{a}$ iff $\sqrt{a} \notin F$ iff $F \models \neg (\exists y)(y^2 =a)$.
To see the right to left direction, note that if $F \models \neg (\exists y)(y^2 =a)$, then $\si$ cannot fix $\sqrt{a}$,
and hence it must map $\sqrt{a}$ to $-\sqrt{a}$, the only other root of $a$ in the algebraic closure.
$\tau$ would need to send $-\sqrt{a}$ to $\sqrt{a}$, since otherwise it must fix $\sqrt{a}$, which means
it would fix $a$ as well, contrary to the assumption on $\tau$.
The left-to-right direction is argued similarly.
 
 Again by results of Ax, the theory of $(F, a_1,\ldots,a_{n-1})$ is determined by 
 \[ k(a_1,\ldots,a_{n-1})^{alg} \meet F
 \]
 This is identical to the same expression for $F'$, so the theories are equal; and likewise for any other $n-1$-tuple.
This completes the proof of Lemma \ref{lem:differonsquare}.
\end{proof}

Note that above we have dealt with embedded finite subsets over distinct pseudo-finite fields. We can get the same
result for a class over the same  field:
 $F,F'$ are elementarily equivalent; so one can find $F''$ and elementary embeddings
 of $F$ and $F'$ into $F''$; of course the images of the $a_i$ will be two different 
$n$-tuples in $F''$, with the same properties as above.
   
We now prove Theorem \ref{thm:psfnotomegabounded}.
   Now recall that, due to the presence of a nondegenerate bilinear form, we have  $L_P$ formulas
  that assert:

\begin{compactitem}
\item   $P$ is linearly independent
\item   
$y$ is the sum of the elements of $P$: in particular,
    $y$ depends on $P$ and for $u \in P$ $y-u$     depends on $P \m \{u\}$.
\end{compactitem}
The  bilinear form is used to state that $y$ is the sum of elements of $P$.
Let $\phi(y)$ be this formula.
We also have an $L_P$ sentence $\psi$
    asserting that the sum of elements of $P$ is a square.  
By  Lemma \ref{lem:differonsquare},
 for any sublanguage $L'$ of $L$
   generated by formulas of arity $n-1$ or less, there exist two $L'$-isomorphic choices of $P$, one satisfying
      $\psi$  and the other not. From this we can easily conclude that
$\psi$ cannot be equivalent to a $k$-RQ formula for any $k$. Therefore the theory is
not $\omega$-RQC.

Notice that our example contradicting $\omega$-RQC is not isomorphism-invariant.
We believe this is not a coincidence.
By Proposition \ref{prop:elem} if there were an isomorphism-invariant $L_V$ sentence
that is not equivalent to a $k$-RQ  one for any $k$, then there could be no elementary time algorithm
for deciding sentences with a fixed number of variables, for any completion of theory of pseudo-finite fields.
We conjecture that there are completions that are ``fixed-variable elementary'', which would disprove
the existence of such a sentence.

%% file: rqcvsonebounded.tex
\section{Embedded  models vs embedded subsets} \label{sec:arity}

We consider on more weakening of RQC, where we  restrict based on the signature for
 uninterpreted
relations, focusing on sentences
dealing with embedded finite subsets rather than  general embedded finite models.
\begin{definition}
We say a theory is \emph{$1$, Monadic-RQC} if every $L_P$ formula is equivalent to a $1$-RQ formula. 
\end{definition}

The motivation is to say whether this provides us with new collapse phenomena -- we will see that it does not.

We begin by extending Proposition \ref{prop:rcqnip}.
\begin{theorem} \label{thm:forall1boundednip} If $M$ has $1$, Monadic-RQC,  then $M$ is NIP.
\end{theorem}

\begin{proof}
If $T$ is not NIP, then by \cite{simontwovarsip},
we can find  an $L$ formula $R(x,y, \vec p)$ such
that in every model there is $\vec c$ such that there are arbitrarily large sets shattered
by $R_{y,\vec c}=\{x ~ | ~ R(x; y, \vec c)\}$. For  short ``$R$ has IP''.
Thus, we can find an infinite order indiscernible sequence
such that every subset of the sequence is $\{x | R(x, b, \vec c)\}$
for  some $b$.

We  find a sequence
$\langle a_i, b_i\rangle_i \in \Nn$ with each   $a_i \neq b_i$,  $\langle a_i, b_i\rangle$ indiscernible
over $\vec c$, and for each $i, j \in \Nn$,
$R(a_i, b_j, \vec c)$ if and only if $i<j$.

We can do this by starting with $\langle a_i\rangle $ indiscernible over $\vec c$, 
then build $b_j$ inductively: the induction step uses that every subset of the sequence is $\{x| R(x, b, \vec c)\}$ for some
$b$.
Then we take a subsequence of 
$\langle a_i, b_i\rangle$ that is indiscernible over $\vec c$.

We now consider several cases.

Case 1: There exists some $b^*$ such that $R(a_i,b^*, \vec c)$
 for all $i$, and
for infinitely many $i$ $\neg R(b_i, b^*, \vec c)$.
Then, by refining, we can arrange that $\neg R(b_i,b^*, \vec c)$ for all $i$ and
$\langle a_i, b_i\rangle$ is indiscernible over $\vec c \cup b^*$.

We can arrange that $\vec c \cup b^*$ has cardinality $0$ mod $4$
via padding.
Now we consider the class $\mathcal C$ of embedded finite subsets where $P$ is interpreted
by
$P_i=A_i \cup B_i \cup \vec c \cup b^*$

We define the $L_P$ sentence $\phi$ to hold if
there exists $\vec c, b^*$ consisting of distinct values,
such that,  letting:

\begin{flalign*}
& A_{\vec c,b*} = \{x \in P ~ | \bigwedge_i x \neq c_i  \wedge x \neq b^* \wedge
R(x, b^*, \vec c)\} &\\
& B_{\vec c, b^*} = P \setminus A_{\vec c, b^*} \setminus b^* \setminus \vec c &\\
& <^A_{\vec c, b^*} = \{ \langle x,y \rangle ~  | x \neq y \wedge x,y  \in A_{\vec c, b^*}
\wedge &\\
& \forall z \in B_{\vec c, b^*} (R(x,z, \vec c) \rightarrow R(y,z, \vec c)\} &\\
& <^B_{\vec c, b^*}  \mbox{ defined similarly as above but with } &\\
& B_{\vec c, b^*} \mbox{ and } A_{\vec c, b^*} \mbox{ swapped } &\\
& Biject_{\vec c,b^*} = \{ \langle x, y \rangle ~  | ~ x \in A_{\vec c, b^*} \wedge &\\
& y  \in B_{\vec c, b^*} \wedge y \mbox{ is } 
 <^B_{\vec c, b^*}\mbox{- maximal such that } R(x,y, \vec c) \} &
\end{flalign*}

then:

\begin{compactitem}
\item $<^A_{\vec c, b^*}$ is a linear order on $A_{\vec c, b^*}$ and similarly
for $<^B_{\vec c, b^*}$.

\item $Biject_{\vec c,b^*}$ is a bijection from  
$A_{\vec c, b*}$ to $B_{\vec c, b*}$

\item for some $b$, $\{x \in A_{\vec c, b^*} ~ | R(x,b, \vec c) \}$ includes exactly
one of any element of $A_{\vec c, b^*}$ and its $<^A_{\vec c,b^*}$ successor, 
while
including the first element and excluding the last element according to 
$<^A_{\vec c, b^*}$.
\end{compactitem}

The important points about $\phi$ are:
\begin{compactitem}
\item $\phi$ is in $L_P$
\item $\phi$   implies, over any embedded finite subset,  that the cardinality of $P$ is $0$ mod $4$. This follows
directly from the definitions and the assumption that $\vec c  \cup \vec b^*$ has cardinality a multiple
of $4$. In particular this implication holds for subsets in $\mathcal C$. 
\item Conversely, for any $P_i \in \mathcal C$, if $P_i$ has cardinality $0$ mod 
$4$ then $|A_i|$ and $|B_i|$ are even, and by choosing $\vec c$ and $b^*$
 correctly --- that is, as above --- we see  that $\phi$ holds.
\end{compactitem}

But $\phi$ cannot be definable by a $1$-RQC formula for embedded finite subsets
in
class $\mathcal C$. Over this class, every $L$-formula is equivalent
(modulo expansion)
to a formula  using $Biject_{\vec c,b^*}$, $<^A_{\vec c, b^*}$,
$<^B_{\vec c, b^*}$
for the correct $\vec c, b^*$, where these are formulas in $L$ extended
with constants. So we have a formula in this language
over a structure consisting of two linear ordered sets with a bijective
correspondence. By a standard Ehrenfeucht-{\fraisse}
argument it is clear that the cardinality of the universe
modulo $4$ cannot be defined in such a family.
This completes the argument for Case 1.

Case 2:  there exists some $b^*$ such that $\neg R(a_i,b^*, \vec c)$
holds for all $i$, and for infinitely many $i$ $R(b_i, b^*, \vec c)$.

This is argued symmetrically with Case 1 above.

Case 3. None of the above.
If $R(b_i,b_j, \vec c)$ and $\neg R(b_j,b_i, \vec c)$ holds for all
$i<j$ (or dually) then $R(x,y, \vec c)$
defines an ordering on the $b_i$.
We let $\mathcal C$ be defined as above but without the $a_i$ or $b^*$, and conclude analogously to Case 1.

Otherwise, by indiscernibility of $b_i$ over $\vec c$,
either   $R(b_i,b_j, \vec c)$ holds for all 
$i \neq j$, or  $\neg R(b_i,b_j, \vec c)$  holds for all $i \neq j$.  
Say the latter.  
Extend the indiscernible sequence $\langle a_i,b_i \rangle_{i \in \Nn}$ by 
adding one more element $(a_\omega,b_\omega)$ maximal in the ordering.
So $R(a_i,b_\omega, \vec c)$ holds for all $i \in \Nn$, 
using indiscernibility, since $b_\omega$ is above $b_i$ for all standard $i$.
But $\neg R(b_i,b_\omega, \vec c)$ holds for all $i$,  again by 
indiscernibility
and $b_\omega$ being above $b_i \in \Nn$.
Hence we are in Case 1, a contradiction.
\end{proof}

We now show that for $1$-RQC, there is no difference between looking at embedded
finite subsets and embedded finite models. That is, our weakening does not make any difference at the level of theories:

\begin{theorem} \label{thm:rqc1bounded} A theory $T$ is $1$, Monadic-RQC iff
$T$ is $1$-RQC.
\end{theorem}

\begin{proof} The interesting direction is assuming $1$, Monadic-RQC
and proving   $1$-RQC.
Inductively it suffices to convert a formula $\phi$ of the form:
\[
\exists x  ~ Q_1(u_1) \ldots Q_n(u_n) ~ \Gamma(x, \vec u, \vec y)
\]
to an active domain formula, where the $Q_i$ are active domain quantifiers and
$\Gamma$ is a Boolean combination of $\sigma$ formulas and $L$ formulas.
We can assume that $x$ only occurs in $L$ formulas.
Let $\Gamma_L(x, \vec u, \vec y)$ be the vector of 
$L$ subformulas of $\Gamma$.
By Theorem \ref{thm:forall1boundednip}, $T$ is NIP.
By \cite{simon-chernikov-two} ``local types are uniformly definable over finite
sets'': 

\medskip

for every $L$ formula $\phi(\vec x; \vec y)$, there is $L$ formula $\delta(\vec y; \vec p)$ such that:
for any finite set $S$, for any $\vec x_0$, there is a tuple $\vec p_0$ such that
$\forall \vec y \in S ~ \phi(\vec x_0, \vec y) \leftrightarrow \delta(\vec y, \vec p_0)$.

\medskip

 This means, in particular, that for $\gamma_i \in \Gamma_L$
there is $\delta_i(\vec u, \vec p)$ such that for every 
finite set $P$  for every $x_0, \vec y_0$ in a model of $T$
there is $\vec p^i$ in $P$ such that
\[
\forall \vec u \in P ~ \gamma_i(x_0, \vec y_0, \vec u) \leftrightarrow \delta_i(\vec u, \vec p^i)
\]
This holds over all finite $P$, so in particular for the active domain of an embedded
 finite model.

We can thus replace $\phi$ with
\begin{align*}
\exists \vec p^1 \in \adom \ldots   \exists \vec p^n \in \adom \\
\vec p_1 \ldots \vec p_n \mbox{ are definers for some } x \mbox{ for } \vec y
\wedge \\
~ Q_1(u_1) \ldots Q_n(u_n) ~ \Gamma'(\vec u, \vec y, \vec p^1 \ldots \vec p^n)
\end{align*}
Here $\Gamma'$ is obtained by $\Gamma$ via replacing $\gamma_i$ with the corresponding
$\delta_i$. The first conjunct inside the existential
has the obvious meaning, that the iff above holds for some $x$ in the role of  $x_0$, with
$\vec y$ in the role of $\vec y_0$.

The first conjunct does not mention the additional relational signature $\sigma$,
and thus can be transformed into a $1$-RQ formula via $1$, Monadic-RQC.

\end{proof}

\myparagraph{Embedded subsets vs embedded models for isomorphism-invariant sentences}
We contrast Theorem \ref{thm:rqc1bounded}
with the situation when we restrict to \emph{isomorphism-invariant} sentences. Recall that this denotes sentences of $L_V$ whose truth value in an embedded finite model for theory $T$ depends only on the isomorphism type of the interpretations of the $V$ predicates. Isomorphism-invariant $L_V$ sentences are ones that represent ``pure relational computation''. Consider the random graph. 
The theory is not NIP, so by Theorem \ref{thm:forall1boundednip}
it is not $1$-RQC, and indeed not even $1$, Monadic-RQC. And  we have mentioned
before that it is $2$-RQC.
In fact, every $L_V$ sentence can be converted to one
that uses only Monadic Second Order quantification over the active
domain of the $V$ predicates. Thus, informally, first-order quantification over the model gives
you the power of Monadic Second Order  active domain quantification over  the relational signature $V$.

We can also see that there are isomorphism-invariant
$L \cup \{G(x,y)\}$ sentences $\gamma$
that are not equivalent to $1$-RQ ones: by using unrestricted
existential quantification to quantify over subsets of the nodes
of $G$, we can express that $G$ has a non-trivial connected component.

But consider an isomorphism-invariant sentence  in $L_P$.
By the observation above, applicable to arbitrary $L_P$  sentences,
 these can all be expressed in restricted-quantifier Monadic Second Order logic $\phi'$ quantifying
only over  unary predicate $P$. By isomorphism-invariance, such a sentence is determined by its truth value
on $P$ lying within an order-indiscernible set, and thus we can rewrite such a sentence to $\phi''$ that
does not mention the graph predicate of $L$ at all:  only equality atoms in it. Such a sentence can only define
the set of finite $P$'s
whose cardinality is in a finite or co-finite set. And such a set of finite $P$'s
can be also be defined in first-order logic over $P$.
Thus there can be a difference between considering unary signatures
and higher-arity signatures for isomorphism-invariant sentences.

%% file: conc.tex
\section{Discussion and open issues} \label{sec:conc}

The main goal of this work was to revisit the evaluation of formulas that refer to a vocabulary $L_V$ mixing the signature of an infinite
structure with a finite uninterpreted relational vocabulary. When this setting was studied in the 90's and the early
200's, the emphasis was on the ``very sunny day'' scenario of  restricted quantifier collapse, where the expressiveness
and evaluation complexity is not so different  -- from the point of view of ``data complexity'' -- from that of evaluation of first-order formulas in a relational vocabulary over
finite models.

In this paper we show a much richer landscape, including:
\begin{itemize}
\item Structures/theories where first-order logic over $L_V$ captures higher-order logic, or beyond. Thus both the expressiveness
and complexity is high.
\item Structures/theories where first-order logic over $L_V$ does not correspond to any reasonable logic over the finite vocabulary, where this can be ``fixed'' by expanding the signature, and cases where it cannot be fixed.
\end{itemize}
We also locate the theory of finite and pseudo-finite fields lives within this landscape, showing that it is extremely expressive.

In each of the approaches to exploring evaluation on wider classes of infinite structures, we have left open many basic questions.

In terms of higher-order collapse,
although we have separated $k+2$  from  $k$-RQC for every $k$,  we
do not have the corresponding separation for $k+1$ and $k$. Thus it is possible that, for example, two consecutive levels of the hierarchy collapse, but the remainder are separated. Our construction from
Subsection \ref{subsec:separate} shows that this would follow from  the corresponding separation
of pure spectra in finite model theory. 
We believe one can construct examples $T_k$ where
the $L_P$ theory is \emph{bi-interpretable} with $k$ higher order logic over $P$, which would provide an equivalence
between separations of RQC levels and expressiveness separations in higher-order logic.

On the issue of persistent unrestrictedness, one major question is whether any model theoretic
criteria can enforce potential RQC. In particular, we do not know if NIP
theories can always be expanded to become RQC. It is known that there are NIP theories (even
stable theories) that cannot be expanded to fall into the known existing classes that imply $1$-RQC  \cite{distalwild}.
We have shown that Atomless Boolean Algebras are persistently unrestricted. We suspect the same argument applies
to another well-known decidable theory, B\"uchi Arithmetic, but we have not verified this.

Our discussion of algebraic examples focused on finite and pseudo-finite fields, and (in the appendix)
on the related topic of vector spaces with a bilinear form into a finite field.
Even for pseudofinite fields and fileds, we do not have results when we restrict to characteristic $0$.

On the issue of impact of the signature, the main question is whether $k$-RQC for $V=\{P\}$, $P$ a unary predicate
$P$ implies $k$-RQC for general $V$.

The results on finite and pseudo-finite fields relate to the broader question of the connection
 between ``tame definability properties'' of a structure such as $k$-RQC and the complexity of the underlying theory.
In  Proposition \ref{prop:elem}, we noted one simple  connection:
roughly speaking,
high expressiveness of isomorphism-invariant $L_V$ sentences implies lower bounds
for  deciding $L$ formulas. In Section \ref{sec:pseudo} we showed that
we can sometimes convert high expressiveness of isomorphism-invariant sentences into lower bounds
on quantifier elimination in the theory, even when the number of variables is not fixed.

We know that, in general, properties like $k$-RQC do not even imply decidability of the theory,  much less
complexity bounds. This follows from the fact that purely model-theoretic properties -- e.g. NFCP, o-minimality, see Section \ref{sec:def} -- with
no effectiveness requirement suffice to get $1$-RQC. Conversely, Example \ref{ex:fcp} can be used to show that theories
with modest complexity of satisfiability (even for fixed number of variables) may not even be $\omega$-RQC. 
Intuitively, a Turing Machine can decide the theory by accessing additional parts of
the structure, while restricted-quantifier formulas have no means to access this structure. One might hope that
adding function symbols to the signature might close this gap, as it does in Example \ref{ex:transient}.
But Theorem \ref{thm:abapersistentlyunbounded}
 shows that there are examples where this tactic  cannot help.


%% file: ack.tex
\myparagraph{Acknowledgments} We thank the referees of LICS for many helpful comments. 
This research was funded in part by EPSRC grant EP/T022124/1.
For the purpose of Open Access, the author has applied a CC BY public
copyright licence to any Author Accepted Manuscript (AAM) version arising from
this submission.

%% file: appendix.tex
\onecolumn
\input{app_bb} 
\input{app_exfcpnotrobust} 
\input{app_boolalgsubalgebra}

\input{app_robust}
\input{app_qe}

\input{app_bilinear}

\input{app_stable}

%% file: app_bb.tex
\section{Proof of Theorem \ref{thm:bbiso}: NIP theories are RQC within an infinite set, and also the proof of Proposition \ref{prop:transient}}
Recall that in the body of the paper we statd

\medskip

 In every NIP theory, there is  a model containing an infinite set  $I$ such that, for embedded finite
models based on $I$, every isomorphism-invariant $L_V$ formulas is equivalent to a $1$-RQ formula.

\medskip

As mentioned in the body, this is a variation of a result in  \cite{baldben}. But the exact statement
of the result in \cite{baldben} is different. Theorem 2.3 of \cite{baldben} states:

\begin{theorem} \label{thm:bbreal} Fix NIP $T$, $M$ a model of $T$, and $I,<$ a dense
linearly ordered set of elements in $M$ that is order-indiscernible: 
the truth value of each formula $\phi(\vec x)$ depends only on the ordering relations among $\vec x$ for $\vec x$ in $I$.

Then any $L_V$  formula is equivalent, over embedded finite models in $M$, to  an RQ formula that can make use of $<$.
\end{theorem}

Now suppose that we have an isomorphism-invariant $L \cup V$ formula $\phi$.

It is well-known that every theory with infinite models has a model with an infinite indiscernible sequence based on a dense linear order.
And isomorphism-invariance implies that any operation can be moved into the set.
Thus to close the gap to Theorem \ref{thm:bbiso}, we need to show that we can eliminate the reference to the indiscernible ordering
$<$ in $\phi'$, instead using $L$ formulas. In the case that $V=\{P\}$ this is easy, since if we have a sentence quantifying over $P$ using only  a linear
order on the domain, and the truth of the sentence depends only on the cardinality of $P$, it  must define a finite or co-finite set of cardinalities.
These can be described without referencing the order.

For general $V$, we provide an argument similar to the one in \cite{simon-chernikov-one} (page 422, below Fact 3.1).

First, suppose that the sequence is not totally indiscernible: that means that for some $L$
formula $\gamma(\vec x)$, we can use $\phi$ to distinguish a tuple $\vec i$ in $I$ from a permutation
of $\vec i$.
Then  for large enough embedded finite models within $I$, the ordering on $I$ is definable with
parameters in the model, by some formula $\lambda(i, i', \vec p)$: see comments below Fact 3.1 in \cite{simon-chernikov-one}.
Consider the $L \cup P$ sentence $\phi''$ that states there is $\vec p$ in the active domain of the embedded finite model such
that $\lambda( x,x', \vec p)$  induces a linear order on the elements in the active
domain and
the sentence obtained from $\psi$ by replacing $<$ with $\lambda(x,x', \vec p)$ holds.
Since $\phi$ was isomorphism-invariant, $\phi''$ is equivalent to $\phi$.

Now suppose the sequence is totally indiscernible. Then we can observe that
the argument in Section 4 of \cite{baldben} goes through to show that
we get a first-order restricted-quantifier formula using only the predicates in $V$ and equality.
The argument is given only for $L_P$, but generalizes to arbitrary
relational $V$. The proof uses only that the trace of an $L$-formula
on a totally indiscernible set is uniformly defined by an equality formula
with parameters in the set. The same holds for a finite indiscernible
set in an NIP theory.

We use a similar approach to prove Proposition \ref{prop:transient}.
Recall the statement
\transient*

\begin{proof}
By way of contradiction, let $\phi'_k$ be a $k$-RQ sentence equivalent to $\phi$ in some expansion $T^+$.
Let $I$ be an order-indiscernible set in a model of $T^+$. Let $\phi''_k$ be obtained by replacing atomic formulas in $\phi'_k$ by a disjunction of ordering
formulas, depending on whether the corresponding atom holds for some -- equivalently all -- tuples with the given ordering.
For $V_0$ an embedded structure in a model $M^+$ of $T^+$, $\phi$ holds in $M^+, V_0$ if and only if $\phi''_k$ holds on $M^+$, $V^I_0$, where
$V^I_0$ is a copy of $V_0$ with active domain in $I$.  Thus, arguing by way of contradiction, we can see that the same equivalence holds in any model of $T$ with an order-indiscernible set, where $<$ in $\phi''_k$
is interpreted as the indiscernible order. Att his point we have $\phi$ equivalent to a higher-order RQ formula, ecept for the use of the ordering.

If $k \geq 2$, we can use higher-order quantification  over the active domain to guess the ordering. If $k=1$ we can
use the argument above to eliminate the use of ordering in favor of $L$ formulas.
\end{proof}

%% file: app_exfcpnotrobust.tex
\section{Proof that Example \ref{ex:fcp} is not  $\omega$-RQC, but is potentially $1$-RQC}  

In the body of the paper we mentioned that 
NIP theories that are not $1$-RQC, or even $\omega$-RQC,  may become  $1$-RQC
when the theory is expanded. And we claimed in Example \ref{ex:transient}  that this phenomena holds
in Example \ref{ex:fcp}, which involves
 an equivalence relation $E(x,y)$ with arbitrarily large finite classes, 
We now give the details.

We let $M^+$  expand $M$ by a discrete linear order $<$ for which each
equivalence class is an interval, and $T^+$ be the theory of this model.
We claim that  $T^+$ is $1$-RQC.
For intuition, consider the formula $\phi \in L_P$ stating
that $P$ contains an equivalence class. In $M$ this was not expressible
by a RQ formula. In $M^+$, $\phi$ is equivalent -- over finite $P$
embedded in $M^+$ --
to the following sentence $\phi'$
\begin{align*}
\exists l \in P ~ \exists u \in P ~ E(l,u) \wedge \phi_{min}(l) \wedge 
\phi_{max}(u) \wedge \\
\forall y \in P [(l \leq y < u)  \rightarrow \exists y' \in P ~ Succ(y,y')]
\end{align*}
where $\phi_{min}(l)$ asserts that $l$ is the minimal element of
its equivalence class, $\phi_{max}(u)$ asserts that $u$
is the maximal element of its equivalence relation, and $Succ(y,y')$
expresses that $y'$ is the successor of $y$ in the linear order.
Recalling that $1$-RQ formulas are built up from atoms of $S$
and $L$ formulas, we see that $\phi'$ is a $1$-RQ formula.

To see the general result, it suffices, by \cite{benediktlibkinominimal}
or \cite{belegradekisolation},  to show
that $M^+$ is $o$-minimal: for every $L$-formula
$\phi(x,\vec p)$ and any $\vec p_0$, $\{x| M^+ \models \phi(x, \vec p_0)\}$
 is a finite union of intervals.

For this it suffices to show that an $L$-formula in variables $\vec x$
is a disjunction of formulas specifying:
\begin{itemize}
\item which of the $x_i$ are equivalent
\item ordering  among the equivalence classes
of the $x_i$ (where the notion of one class being above another is 
the obvious one)
\item ordering among the $x_i$ in the same class,
\item for each $x_i$ and $x_j$ that are neighbors in the ordering lying
within the same class,
\item information (lower bounds or exact bounds) on the distance of the
lowest $x_i$ in the class from the minimal element of the class,
and similarly for the highest $x_i$ in a class.
\end{itemize}

Given this statement, if we fix values for all but one free variable $x$, and
fix a formula as above, 
the set of values will be a single interval. The $o$-minimality result will then follow.

The quantifier-elimination result above is proven by a standard induction.

%% file: app_boolalgsubalgebra.tex
\section{Alternative proof of Claim \ref{clm:fv}: $1$-RQC over subalgebras
of an Atomless Boolean algebras}

Let $T$ be the theory of Atomless Boolean Algebras. In the body of the paper
we claimed that for this theory, every isomorphism-invariant sentence
was equivalent to a $2$-RQ sentence.
In the body we argued that it suffices to prove Claim \ref{clm:fv}:

\medskip

For every $L_V$ formula $\phi$ there is a $1$-RQ formula
$\phi'$ that is equivalent to $\phi$ over embedded finite models in which 
the active domain of the reduct to $V$ is a subalgebra.

\medskip

We gave a model-theoretic proof in the body, and here we give an alternative algorithmic argument.

As usual it suffices to eliminate one unrestricted quantifier at a time, converting:

$\phi=\exists x ~ Q(u_1) \ldots Q(u_k) ~ \Gamma(x,\vec u, \vec y)$

Where $Q(u_i)$ are active domain quantifiers and $\Gamma$ is a Boolean combination of $L$ formulas and $S$ formulas.
We ignore the $S$ formulas in $\Gamma$ in the analysis below.

Atomic formula in $\Gamma$  are of one of the forms 
$x \cap \sigma_1(\vec u) \cap \sigma_2(\vec y) \neq 0$, $x^c \cap \sigma_1(\vec u) \cap \sigma_2(\vec y) \neq 0$
where $\sigma_1, \sigma_2$ are terms built up using
$0,1$ and the usual Boolean operators. We will assume that $\Gamma$ covers all terms in these variables.

A Boolean algebra term $\sigma_1$ in variables $\vec x$ is contained in Boolean algebra
term $\sigma'_1$ if we have containment for every valuation of $\vec x$.

Fixing the variables as above, consider any element $x$ of the algebra. 
For a given $\sigma_2$,  and an element $x_0$ and vector of elements
$\vec y_0$ the \emph{$\sigma_2$ representation} of  $x_0, \vec y_0$ is the pair $(e_1, e_2)$
where $e_1$ is the union of all atoms 
$u$ in the subalgebra such that $x_0 \cap  u \cap \sigma_2(\vec y_0) \neq 0$
and $e_2$ is the union of all atoms $u$
in the subalgebra such that $(x_0)^c \cap  u \cap \sigma_2(\vec y_0) \neq 0$.
The \emph{full representation} of $x_0, \vec y_0$ consist of representatives for each $\sigma_2$. We index
a full representation as $\vec w$, where $w_{x, \sigma_2}$ represents the first components of the pair for
$\sigma_2$ and $w_{x^c, \sigma_2}$ represents the second component of that pair.

A vector of pairs $p_{\sigma_2}$ where  $\sigma_2$ range over terms in the appropriate variables, is said
to be \emph{$\vec y_0$ realizable} if there is some $x_0$ such they are the full representation of $x_0, \vec y_0$.

\begin{proposition} \label{prop:realizable}There is a $1$-RQ formula $\kappa(\vec w, \vec y)$ that holds if and only if
$\vec w$ is $\vec y$ realizable.
\end{proposition}

Given the proposition, we rewrite $\phi$ to 
\[
\phi'= \exists \vec w \in \adom ~ \kappa(\vec w, \vec y) \wedge
 Q(u_1) \ldots Q(u_k) ~  \Gamma'(\vec u, \vec w)
\]
where in $\Gamma'$ $x \cap \sigma_1(\vec u) \cap \sigma_2(\vec y) \neq 0$ is replaced with  a formula stating
$\sigma_1(\vec u)$ contains an atom in $w_{x, \sigma_2}$, while
$x \cap \sigma_1(\vec u) \cap \sigma_2(\vec y) \neq 0$ is replaced with  a formula stating
$\sigma_1(\vec u)$ contains an atom in $w_{x^c, \sigma_2}$.
 
We now give the proof of Proposition \ref{prop:realizable}.
We say a vector $\vec w$ of pairs of elements as above  is \emph{consistent for
$\vec y$} if:
\begin{itemize}
\item (Complementarity) for every index $\sigma_2$, for every atom $a$ of the subalgebra
such that $\sigma_2(\vec y) \cap a$ is non-empty, $a$
intersects either  $w_{x, \sigma_2}$ or $w_{x^c, \sigma_2}$.
Further, for any non-maximal $\sigma_2$ and any $y_i$ not mentioned
in $\sigma_2$,  $a$
intersects either  $w_{x, \sigma_2 \cap y_i}$ or $w_{x^c, \sigma_2 \cap (y_i^c}$.

\item (Containment) If  $a$ intersects $w_{x, \sigma_2}$ 
then
$\sigma_2(\vec y) \cap a$ is non-empty
and if   $a$ intersects $w_{x^c, \sigma_2}$ then
$\sigma_2(\vec y) \cap a$ is not all of $\sigma_2(\vec y)$.

\item  (Monotonicity) if $\sigma_2$ is contained in $\sigma'_2$, then
$w_{x, \sigma_2}$ is contained in $w_{x, \sigma'_2}$ and similarly
for $x^c$.
\end{itemize}

%

Clearly consistency can be expressed as a $1$-RQ formula.
It is also easy to see that if vector is $\vec y$ realizable then it is 
$\vec y$ consistent.
We finish the argument by showing
that if a vector $\vec w$ is $\vec y$ consistent then it is $\vec y$ realizable.
Given a $\vec y$ consistent $\vec w$, for each $\sigma_2$ of the maximal
length we let $x_{\sigma_2}$
be the union of the following elements:
\begin{itemize}
\item $u \cap \sigma_2(\vec y)$, for  each atom $u$  of the subalgebra that 
intersects 
$w_{x, \sigma_2}$ but does not intersect $w_{x^c, \sigma_2}$, 
\item for each atom  $u$ of the subalgebra that intersects
both $w_{x, \sigma_2}$ and $w_{x^c, \sigma_2}$, an element of
the atomless Boolean algebra that is a non-zero  proper subset
of $\sigma_2(\vec y) \cap u$. By the the (Containment)  condition,
for each atom $u$ that intersects $w_{x, \sigma_2}$ 
$\sigma_2(\vec y) \cap u$ is non-empty, so such a subset exists.

\end{itemize}
We let $x_0$ be the union of $x_{\sigma_2}$ over all maximal-length $\sigma_2$.
We claim $x_0$ is the desired realization.

Fixing an arbitrary $\sigma_2$, we need to argue that the $\sigma_2$ 
representation
of $x_0$ is the pair $w_{x,\sigma_2}, w_{x^c, \sigma_2}$.

We first show that   $w_{x,\sigma_2}$ is correct.
In one direction, suppose $u$ is an atom of the subalgebra
for which $x_0 \cap u \cap \sigma_2(\vec y)$ is non-empty.
We argue that $u$ intersects $w_{x,\sigma_2}$. If $\sigma_2$ is
of maximal length, this is true by definition. In the general case
we can choose some maximal $\sigma'_2$ that $x_ \cap \sigma'_2(\vec y)$ is
non-empty and then argue via (Monotonicity).

In the other direction, suppose that  $u$ intersects $w_{x,\sigma_2}$.
We need to argue that  $x_0 \cap u \cap \sigma_2(\vec y)$ is non-empty.
By repeated applications of (Complementarity), we can find a $\sigma'_2$
of maximal length such that $u$ intersects $w_{x, \sigma_2}$.
Now by definition $x_0 \cap u \cap \sigma'_2(\vec y)$ is non-empty, and
the result follows.

We now give the argument that $w_{x^c,\sigma_2}$ is correct.
This will not be quite symmetrical, since we did not form $x^c$ explicitly.

In one direction, suppose $u$ is an atom of the subalgebra
for which $(x_0)^c \cap u \cap \sigma_2(\vec y)$ is non-empty.
We need to argue that $u$ intersects $w_{x^c,\sigma_2}$. 
Let us consider the case where $\sigma_2$ is maximal.
We know that $u \cap \sigma_2(\vec y)$ is not contained in $x_0$.
Thus $u$ could not have fallen under the first condition in the construction of
$x_0$.
This means that $u$  either does not
intersect
$w_{x, \sigma_2}$ or intersects $w_{x^c, \sigma_2}$.
But if $u$ does not intersect $w_{x, \sigma_2}$, by (Complementarity) it
must intersect $w_{x^c, \sigma_2}$, and hence we can conclude
that $u$ must intersect $w_{x^c, \sigma_2}$ as required.
In the case that $\sigma_2$ is not maximal, we apply (Monotonicity).

In the other direction, suppose $u$ intersects $w_{x^c,\sigma_2}$.
We argue that  $(x_0)^c \cap u \cap \sigma_2(\vec y)$ is non-empty.
By repeatedly applying (Complementarity), we can assume $\sigma_2$ is of
maximal length. Suppose by way of contradiction that $u$ were contained
in $(x_0) \cap \sigma_2(\vec y)$. Then $u$ must have fallen into the
first case of the construction of $x_0$ for $\sigma_2$. 
Thus  $u$  
intersects
$w_{x, \sigma_2}$ but does not intersect $w_{x^c, \sigma_2}$, a contradiction.

%% file: app_robust.tex
\section{Proof of Proposition  \ref{prop:transient}: isomorphism-invariant sentences that can not be  converted to
restricted-quantifier form,  cannot become restricted-quantifier in an expansion}

Recall the statement:

\medskip

If a theory $T$ has an isomorphism-invariant $L_V$  sentence $\phi$ that is not
equivalent to a $k$-RQ  sentence relative to $T$, then the same will be true in any expansion
of $T$.

\medskip

\begin{proof}
We first deal with the case of $k>1$, fixing a $\phi$ and considering some expansion $T^+$ of $T$ where
$\phi$ is equivalent to a $k$-RQ $\phi'$.
By considering an order indiscernible set in some model of $T^+$, $\phi'$ can be rewritten to have only the $V$ predicates and the ordering
relation on the indiscernibles. This ordering can be guessed, since $k>1$. Thus $\phi$ can be rewritten to have only the $V$ predicates,
for all models of $T^+$. An equivalence $M \models \phi \leftrightarrow \phi'$ over finite embedded models does not depend on the model of $T$,
since a counterexample would be witnessed by an $L$ sentence and $T$ is complete. Hence
the equivalence of $\phi$ to a $k^{th}$ order $V$ sentence holds in models of $T$ as required.

Now assume $k=1$, and fix   $\phi$. The only difference in the argument will be that we do not get rid of the ordering by guessing
it. Over models of $T^+$ $\phi$ is equivalent to a sentence using only the $V$ predicates and the  ordering
relation on the indiscernibles, thus $\phi$ can be rewritten using only the $V$ predicates and an additional arbitrary ordering on the finite
model. The same will be true in all models of $T$, by the observations above.
If $T$ has a totally indiscernible set $I$ -- one where the truth values of $L$-formulas $\phi(i_{j_1}, \ldots i_{j_k})$ is invariant
under permutations -- then over this set, the ordering can be removed, and thus $\phi$ is equivalent to a sentence using only  $V$ atoms
on models of $T$, a contradiction of the hypothesis that $\phi$ was not $1$-RQ in $T$.
If $T$ does not have such a set, then there is an order indiscernible $I$ where the  ordering over $I$ is $L$-definable with parameters 
from $I$: see also the proof of  Theorem \ref{thm:bbiso} in the other appendix.
 The truth value of $\phi$ will not depend on  which parameters are chosen, so they can be existentially quantified.
Thus $\phi$ is equivalent to a $1$-RQ formula over models of $T$, as required.
\end{proof}

\section{Details in the proof of Theorem \ref{thm:abapersistentlyunbounded}: Atomless Boolean Algebra is persistently unrestricted}
Recall the statement:

\medskip

No expansion of a theory of an infinite
Boolean Algebra is $\omega$-RQC.

\medskip

\begin{proof}
We give the argument for $\alba$, but the reader will see that atomlessness
is not crucial.
We start by considering one particular extension $\alba^<$, which we describe
as the complete theory of a particular model. In this model  we have a distinguished
partition of $1$ into countably many elements $E_i$ ordered by a discrete linear order $<$.
 We look at the Boolean Algebra generated by these sets.
Each element of the subalgebra is of the form  $E_S=\bigcup_{i \in S} E_i$< 
where $S$ is a subset of
of the natural numbers.
We can identify $E_S$
 with the infinite bit strings over $\nats$ where there is a $1$ on the $i^{th}$ bit exactly
when $i \in S$. We order two elements of the subalgebra via the lexicographic
ordering on   the corresponding strings.
Let $\phi_{union}$ be the $L_P$ sentence stating that there is an element that is the union of all the other elements.
We claim that $\phi_{union}$ is   not finitely redeemable within $\alba^<$.
For a given $k$ we create sets $U_k$, $U'_k$ by letting $U_k$ include $E_i$ for $i$ in
a  set $S_k$ of natural numbers of size $k$ and also $E_{S_k}$.
$U'_k$ consists of $E_i$ for $i \in S_k$ and also $E_{S'_k}$ where $S'_k=S_k\cup \{i'\}$ where $i'$ is above
every element in  $S_k$. The embedded finite subset $U_k$ satisfies $\phi_{union}$,
while $U'_k$ does not.
Let $f_k$ be the identity on $E_i$ for $i \in S_k$ and map $E_{S_k}$ to $E_{S'_k}$. Then
$f_k$ is a bijection from $U_k$ to $U'_k$ which preserves all atomic formulas of arity
at most $k$.

The proof closes as described in the body of the paper.
We recall the characterization of everywhere Ramsey expansions from \cite{definabilitypatterns}, which we denote as $(\dagger)$:

\medskip

For every $M \models T$  every $k$, and every set of $k$-tuples $A$ in $M$, there exists a
an elementary extension $(M^*, A^*)$ of $(M, A)$ containing a copy $M'$
of $M$ such that $A^* \cap M'$
is  definable by an $L(M)$ formula without parameters.

\medskip

In \cite{definabilitypatterns}, Section 5.6, it is shown roves that $\alba^<$ satisfies $(\dagger)$.
Now consider an  expansion $M^*$ of a model of $\alba$. 
We can further expand $M^*$ to a model of $\alba^<$ by choosing the ordering
appropriately. By a compactness argument, we can assume that $M^*$ is countably saturated: any countable collection of formulas
with parameters from the model that is consistent is satisfied in $M^*$.
Let $U_n, U'_n, f_n: n < \omega$  be the witnesses above that $\phi_{union}$ is not finitely redeemable in $\alba^<$.
We will argue that there is  a sequence $Q_n, Q'_n, f'_n$ witnessing that $\phi_{union}$ is still not finitely redeemable in $M^*$.
Fixing $n$, it suffices, by countable saturation in $M^*$, to show that for any fixed $\phi_1(\vec x) \ldots \phi_n(\vec x)$  there
 are $Q_n, Q'_n, f'_n$ such that $f'_n$ preserves each $\phi_i:i< n$, $(M^*,Q_n) \models \phi_{union}$, and $(M^*, Q'_n) \models \neg \phi_{union}$.
By $(\dagger)$ there is a copy $M_0$ of $\alba^<$ where each $\phi_i$ intersected with the domain of $M_0$ is $L(M)$-definable.
The copies of $U_n, U'_n$ in $M_0$ can be taken as witnesses: $f_n$ will preserve each 
$n$-ary definable set, thus it will preserve each $\phi_i$ for
tuples in the copy.
\end{proof}

%% file: app_qe.tex
\section{Proof of Theorem \ref{thm:pseudofiniteqe}}
We recall the statement of Theorem \ref{thm:pseudofiniteqe}:

\pseudofiniteqe*

Let $x=(x_1,\ldots,x_n)$ be an $n$-tuple of variables, $y$ another.   A {\em basic} formula is one of the form
 \[(\exists y) (f_m(x) y^m + \cdots + f_0(x) \wedge f_m(x) \neq 0 \wedge \theta(x,y)) \]
  where $\theta(x,y)$ is a quantifier-free formula in the language of rings.

Notice that such a formula is algebraically-constrained, and does not use parameters.
Our first step is to reduce to a Boolean combinations of such formulas:

  \begin{proposition} \label{prop:bcbasic} Any formula $\phi(x)$ in the language of field  is  equivalent, uniformly over finite fields,  to a Boolean combination of 
basic formulas.   \end{proposition}

\begin{proof}   By absoluteness of arithmetic statements, we may assume $2^{\aleph_0} = \aleph_1$ 
We will work model-theoretically. It suffices to show that:

\medskip

if $M,N$ are saturated models of the theory of pseudo-finite fields of cardinality $2^{\aleph_0}$,
$a \in M^n, b \in N^n$, 
and $a,b$ satisfy the same basic formulas, then there exists an isomorphism $f: M \to N$ with $f(a)=b$. 

\medskip

Let $A_0$ be the field spanned by $a$, and  $B_0$ the field spanned by $b$.  Then we have an isomorphism $f_0: A_0 \to B_0$.  

Let $A_1,B_1$ be the   perfect hulls (see \cite{zoenotes}) within $M,N$ respectively.  Extend  $f_0$ to $f_1:A_1 \to B_1$:    if $x \in A_1$, then $x^q=a_0 \in A_0$
for some power $q$ of $p$.  Then $M \models (\exists x)(x^q=a_0)$ and it follows that $N \models (\exists y)(y^q = f_0(a_0)$.
Moreover, given the choice of $q$, $y$ is unique.  Map $x$ to $y$; it is easy to check that this is well-defined and an isomorphism $A_1 \to B_1$.

We now claim:

\medskip

If $K$ is any finite field extension of $A_1$,  $f_1$ can be extended to have domain $K$.  

\medskip

We prove the claim.
 For this claim, we may identify $A_1$ with $B_1$ via $f_1$, so we can assume $f_1$ is the identity restricted to ${A_1}$.  If the base field is perfect, then every finite extension is separable. By the theorem of the primitive element, $K=A_1[c]$ for some element $c$.  The minimal polynomial of $c$ over $A_1$
can be written as $f_m(a^{1/q} ) y^m + \cdots + f_0(a^{1/q}) =0$, with $f_i$ polynomials  over $\Zz$.  So 
$f_m^q(a) y^{qm} + \cdots + f_0^q(a) =0$.   Since $b$ satisfies the same basic formulas as $a$,
we can find $d \in N$ with $f_m^q(b) y^{qm} + \cdots + f_0^q(b) =0$, and by basic algebra we have $A_1[c] \cong B_1[d]$.  

Let $A,B$ be the  relative algebraic closures of $A_1,B_1$ in $M,N$ respectively.   We will argue now that $f_1$ extends to $f: A \to B$, using Koenig's tree lemma.

Let $K_1 \leq K_2 \leq \cdots $ be a chain of finite field extensions of $A_1$, with $\union_m K_m = A$.  
Consider the tree whose $n^{th}$ level is the (finite) set of  extensions of $f_1$ to  $K_n$.))        Similarly $f \inv $ extends to a map $g$ from $B$ to $A$.
Any field homomorphism $gf: A \to A$ over $A_1$  must be surjective, since the restriction of $gf$ to $A \meet N$, with $N$ any finite Galois extension of
$A_1$, is surjective.  Hence $gf$ is  surjective, so $f$ is surjective, so $f: A \to B$ is an  isomorphism.  This completes the proof of the claim.

Returning to the proof of the proposition, by \cite[Theorem 1]{ax} $(M,a) \cong (N,b)$.  This completes the proof of the proposition.
 \end{proof}

 The proof above also shows that any formula is equivalent to a disjunction of conjunctions of a single basic formula and a negated basic formula: the proof of the claim   uses just one basic formula,
as does the dual direction -- i.e. the same claim used from $N$ to $M$.

Now let us return to the proof of Theorem \ref{thm:pseudofiniteqe}. We note that  a negated basic formula is equivalent to a disjunction of basic formulas with parameters.
For example consider
   $\neg (\exists y)(y^2=x)$. This is equivalent to  $(\exists y)(y^2 \cdot c=x)$, where
   $c$ is any element that is not a square.  See e.g. the remarks in \cite{louremarkax}.

Also note that a conjunction of basic formulas is an algebraically-constrained formula. This completes the proof of Theorem \ref{thm:pseudofiniteqe}.

%% file: app_bilinear.tex
\section{Analogs of Atomless Boolean Algebras with bilinear forms}

We know that NIP theories  may not have§ $1$-RQC or even $\omega$-RQC (see Example \ref{ex:fcp}).
But they have $1$-RQC
when restricting interpretations of $V$ predicates which lie within an order-indiscernible set. Hence
isomorphism invariant sentences are well-behaved: contained in order-invariant FO.

Atomless Boolean Algebras are an example that is similar to NIP theories, to some extent:
they do not have $1$-RQC or even $\omega$-RQC, but for very specialized interpretations  they
have $1$-RQC. Hence isomorphism-invariant $L_P$ sentences (Monadic uninterpreted signature) are still equivalent to $1$-RQ sentences.
In this case, the well-behaved instances are not indiscernible bur rather have certain closure
properties: namely they are subalgebras.

We show similarly  phenomena for another basic algebraic example, a vector space with a bilinear form.

Assume $T$ consists of an infinite Abelian  group structure $(G,+)$ of finite 
exponent $p \neq 2$ and a nondegenerate symmetric
 bilinear form  $\bilin( \textunderscore , \textunderscore ): G^2 \to GF(p)$.

We start by  showing that $T$ does not have $1$-RQC, and does not even have it when we restrict
to an indiscernible set.  We first note that:

\begin{proposition}
In any model there an infinite $I$ that is indiscernible
and  such that
$\bilin(d,d) =a \neq 0$ for each $d$
and $\bilin(d,c)=0$ for distinct $d,c$. 
\end{proposition}

\begin{proof}
We show by induction that a finite set of mutual orthogonal elements with 
$\bilin(d,d) \neq 0$ 
cannot be maximal.  Letting $A$ be the span of such a set, it is clear
that $A^\perp \cap A=\{0\}$. We can also see
that $A^\perp$ is infinite:
considering
the standard map $g$ from $V$ to the dual $A^*$ of $A$, we have
 $A^\perp$ is the kernel of $g$. If $A^\perp$ were finite, since $A^*$
is finite, we would have that $V$ is finite.
But choosing a non-zero element of $A^\perp$,
we have contradicted maximally.
Given that we have an infinite set with  $\bilin(d,d) \neq 0$
we can take a subset where the value of $\bilin(d,d)$ is constant.
\end{proof}

We will consider a $J$ interpreting $P$ ranging over finite subsets of
$I$.
We refer to this as an  \emph{orthonormal interpretation} below.

The
 span $V$ of such a $J \subset I$ interpreting $P$
is $L_P$ definable:
the span is the same as the orthogonal complement of the orthogonal
complement:
$e$ is in $V$ if and only if $\forall d$ if $\bilin(d,c)=0$ for each $c \in J$
then $\bilin(d, e)=0$.

Given an element $e$ of the span of $J$, consider the set
$S_e= \{i \in I ~ \mid ~ \bilin(e,i) \neq 0 \}$. This is a subset
of $J$, and it is easy to see (using the properties of $I$)
that every subset $S$ of $J$ is
$S_e$ for a unique $e$ in the span, namely the sum of the elements
of $S$. Thus we can quantify over subsets of $J$ using quantification
in $L_P$.

We can also define the parity of a subset $S$ of $J$ modulo
$p$,  and in particular, the parity
of $J$ itself.   Namely, if $S$ is represented by its sum
$e$ in the span, then  $|S| = a \cdot \bilin(e,e)$ modulo $p$.  

Using this
we  see that $1$-RQC fails even when restricting interpretations to lie inside this indiscernible set. 
 As the trace on $J$ of the first-order structure on $G$ is just pure 
equality, parity cannot be defined by any first-order logic formula on $J$.

 We now consider restricting interpretations to  finite  \emph{non-degenerate subspaces} of a model
$M$ of the theory: those for which the form restricted to the subspace is non-degenerate.
Any finite set $P_0$ has a finite superset $P'_0$ that is in this class: for  any
$p \neq 0 \in P_0$ with $\bilin(p, q)=0$, for each $q \in P$, we throw into
$P_0$ an element $q'$ such that $\bilin(p, q') \neq 0$. We then close under addition to
get a subspace.

We claim that for interpretations of $P$ lying within this class, the theory has $1$-RQC.
We will use the well-known fact that as a vector space
$M$ is the direct sum of $P$ and $P^\perp$. And the bi-linear form structure
on the direct sum is determined by its values on the two components.
Thus by a standard Fefferman-Vaught argument, any sentence $\phi$ in $L_P$
decomposes into a Boolean combination of sentences $\phi_P$ quantifying over $P$ and sentences $\phi^\perp_P$ over
$P^\perp$, uniformly in $P$.
But  $P^\perp$ is also a model of the same theory, and thus by completeness its theory
does not depend on $P$. Thus we can reduce to a $1$-RQ sentence.
The argument for open formulas is similar.

We combine the above two arguments above to conclude:

\begin{proposition} If we restrict interpretations  of $P$ to be  orthonormal, then every $L_P$ sentence
is equivalent to a $2$-RQ one. Since every interpretation can be mapped to an orthonormal one,
we have every $L_P$ isomorphism-invariant sentence is equivalent to a $2$-RQ one.
\end{proposition}

\begin{proof} Fix $\phi$ in $L_P$. We let $\phi^*$ be such
that for any  $P$ in the class, $\phi$ evaluated over $P$
is the same as $\phi^*$ evaluated over $P^*$ equal to the span of $P$.
$P$ is definable within its span, so $\phi^*$ can easily be created.
There is a  $1$-RQ formula
$\phi'$ that is equivalent to $\phi^*$ over non-degenerate subspaces.
In the orthonormal case, the span of $P$ is a non-degenerate subspace,
thus $\phi'$ evaluated over $P^*$ is equivalent to $\phi^*$ over $P^*$.
However the span of $P$ is clearly restricted-quantifier second-order definable over $P$.
Thus replacing $P^*$ in $\phi'$ by the RQ second-order definition shows
that we have $2$-RQC when restricting to embedded finite models in  this class.
\end{proof}

Note that this contrasts with the case of pseudofinite fields, where isomorphic-invariant
sentences can be extremely complex: see Theorem \ref{thm:pseudofiniteposexhaustive}.
 
Continuing  the analogy with Atomless Boolean Algebra, we show that the
theory does not have $\omega$-RQC: obviously, using queries that are not isomorphism-invariant.

 We proceed in the same way as for Example \ref{ex:fcp}, Atomless Boolean Algebra, and pseudofinite fields, by providing  an $L_P$ sentence $\phi$ such  that:

\medskip

for each $k$, there are $P_k,P'_k$ and a function $f_k$ taking $P_k$ to $P'_k$
such that: $(M,P_k) \models \phi_{contained}$, $(M,P'_k) \models \neg \phi$,
but  $f_k$ preserves all formulas with at most $k$ variables.

Following the terminology given in the earlier appendix, we  say that the formula $\phi$ is
not finitely redeemable.

For simplicity, we will assume $p=2$ in the argument below.


Consider the sequence of subsets $C_k$ and $D_k$ such
that $C_k= \{x_1,\ldots,x_k\}$ with $\bilin(x_i,x_j)=0$ for each $i,j \leq k$, and
the $x_i$ are linearly independent. The fact that $G$ is infinite implies that
such a $C_k$ exists. Note that  $\bilin(x_i, x_i)=0$, in contrast to the assumption
on the set $I$ considered above.  Let
$D_k  =\{x'_1,\ldots,x'_k\}$ with $x'_i=x_i$ for $i \leq k-1$ and
$x'_k=-\Sigma_{i \leq k-1} x_i$. Note that we still have $\bilin(x'_k, x'_j)=0$ for $j <k$.

Note that in $D_k$ there is an element that is in the span of the remaining
elements, but this is not the case in $C_k$. It is thus easy to distinguish
them in $L_P$.

Fixing an arbitrary bijection $f_k$ from $C_k$ to $D_k$, we next claim that $f_k$ preserves
each \emph{atomic} formula of size at most $k-1$. Atomic formulas of the form
$\bilin(\tau_1, \tau_2)=0$ for arbitrary terms $\tau_1, \tau_2$ will be true when bindings of
variables 
to either $C_k$ or  $D_k$ are applied. 
While formulas
of the form $\bilin(\tau_1, \tau_2)=1$ are the negation of the equality above, given that $p=2$.
Atomic formulas of the form $\tau_1=\tau_2$ will be preserved when the sum of the
sizes of $\tau_i$ is at most $k-1$.

We then claim
that for any fixed set $J$ of $L$ formulas $\phi$, there is $k$ such that
$f_k$ preserves the truth of all formulas in $J$. This follows from quantifier-elimination
for the theory, which is argued in \cite{ch}.
This completes the argument that the theory is not finitely redeemable, and hence not $\omega$-RQC.

Although the results above indicate a similarity between this example and Atomless Boolean Algebras,
we do not know if this example is persistently unrestricted.

%% file: app_stable.tex
\section{Separating RQC levels with NIP theories}
In the body of the paper we proved that there are theories that are $k+2$-RQC
but not $k$-RQC. The argument also shows that one can separate $k+1$ from $k$
if one can separate the hierarchy of pure spectra. 
But the examples we get are not NIP.
We mentioned that for NIP theories we can get separation of $2$-RQC and $1$-RQC.
This is more evidence that  imposing NIP does not tell us much about higher RQC:
it does not imply even $\omega$-RQC  (Example \ref{ex:fcp}), and  it does not imply
that the $k$-RQC hierarchy collapses.

We now give a separating example. It will be decidable, NIP (in fact stable, which implies NIP),
and $2$-RQC, but not $1$-RQC. We do not have NIP examples that
are  $k+2$-RQC but  no $k$-RQC  for $k>1$.

Let $I_n=[a_n, b_n]: n \in \omega$ be a partition of the natural numbers into 
intervals, where $b_n-a_n =n$. Let $g(n)= \Sigma_{i \in I_n} i$.

Consider a signature with binary relations $E$, $F$, and unary relations $\mymax$
and $\mymin$,
and let $M$ be a structure in which $E$ is an equivalence
relation with equivalence
classes $C_n$ for each $n$, where $C_n$ has size $g(n)$.
 $F$ is a binary relation such that
for each $n$, $F$ defines a bijection on $C_n$ such that
there is one $F$-orbit of size $i$ for each $i \in I_n$.
$\mymin$ contains, for each $E$ equivalence class $C_n$
 exactly one element of  $c \in C$, where 
$c$ must be in an orbit of size $a_n$. $\mymax$ also contains
one element in each $C_n$, but this time one that lies in an orbit
of size $b_n$.

Formulas in the base language $L$
can express whether the $x_i$ are equivalent,
 and  if they
are equivalent whether or not they are a fixed $F$-distance from each other
or from the unique $\mymin$ and $\mymax$ elements in  their equivalence-class.
One can thus show directly that all types are definable, and thus that
$M$ is stable.

The statement $\phi(P)$ that $P$ takes up a whole equivalence class is
clearly not  definable by a restricted-quantifier formula, so $M$ is not
$1$-RQC.

We now argue that $M$ is  $2$-RQC.

By a simple ultraproduct argument it suffices to show the following:

Let $M'$ be an ultrapower of $M$ via a  non-principal ultrafilter.
In $M'$ say a subset of the domain is hyper-finite if it is an ultraproduct
of  finite sets.
Let $P_0$ and $P'_0$ be hyper-finite sets in $M'$ such that the expansion
of $M$ with $P$ interpreted  by $P_0$ and the expansion of $M$ with $P$ interpreted by $P'_0$   agree on all restricted-quantifier second order
formulas, where the quantifiers are interpreted to range over hyper-finite
subsets of the interpretations of $P$. Then the expansions with $P_0$ and $P'_0$ agree on all $L_P$ formulas.

To show the conclusion, it suffices to construct a back-and-forth.
Suppose we have a partial
isomorphism $h$ in $M'$ between hyper-finite sets $P$ and $P'$,
and an $x$ in $M'$. We will show that we can extend $h$ to include $x$.
Let $C_x$ be the $E$-equivalence class of $x$ in $M'$. The interesting case
is when $C_x$ is hyper-finite but not finite, so we assume this below.
If $C_x$ does not contain an element of $P$, then we can map $x$
to an arbitrary infinite equivalence class that does not contain an element
of $P'$.  So assume that $C_x$ contains an element $p_0 \in P$, and
let $C'$ be the equivalence class of $h(p_0)$.
We will also assume that $P$ contains the unique element satisfying
$\mymin$ and the unique element satisfying $\mymax$ in $C_x$, since these can
be referenced with restricted-quantifier formulas.

If $x$ has finite $F$-distance $k$ from some $p \in P$, then
we can map it to the unique element $x'$ that is $k$ away from $h(p)$ in the same direction, 
and it is clear that this extends $h$ to a partial isomorphism.
So we can further suppose $P$ is disjoint from the $F$-orbit containing $x$.

First consider the case where for arbitrarily large $k$, $C_x$  has an orbit that
contains an element $p_1$ of $P$ but does not contain any elements in an $F$-interval of size
$k$ around $p_1$, where by an $F$-interval we mean a set $o, F(o)\ldots F^k(o)$.  Then the same is true for $P'$ in  $C'$, since we can talk about the interval
around $h(p_1)$ with first-order restricted quantifiers.
Thus (by saturation) $C'$ will contain a hyperfinite $F$-interval that contains  no elements of $P$, and
hence in particular contains an orbit with no elements of $P$. We can map $x$ into that orbit.

Now consider the case where there is a bound $k_0$ on the size of  a $P$-free
interval in any of the orbits containing a $P$ element within $C_x$. That is, in orbits
which are not $P$-free, the $P$
elements are very dense.  Then there is a second order
restricted-quantifier formula $\phi(p, p')$ such that for each $p \in P$ $\phi$ holds of all
the $p' \in P$ that are in the $F$-orbit of $p$. The formula quantifies over a partial function
taking elements of $P$ to the nearest element of $P$ in terms of $F$-distance.
From this we can  find a restricted-quantifier formula that defines the $F$-orbit of $p$.
Finally, we can use this latter formula
to obtain a restricted-quantifier second-order logic formula enforcing that the $F$-orbit
closure of $C_x$ is not all of $C_x$. Informally the formula says that we ``skip a cardinality''.
We state that there are $p,p' \in P$ such that the cardinality of the
$F$-orbit  of $p'$ is larger than the cardinality of the $F$-orbit
 of $p$, and there is no $p'' \in P$ such that the $F$-orbit  of $p''$
has cardinality strictly between that of $p$ and that of $p'$.
Here the restricted second-order quantifiers are existential and the witnesses
can be taken to be hyper-finite. Thus the formula  will also hold in $P'$, which
allows us to map $x$.

%% file: main.bbl
\begin{thebibliography}{10}
\providecommand{\url}[1]{#1}
\csname url@samestyle\endcsname
\providecommand{\newblock}{\relax}
\providecommand{\bibinfo}[2]{#2}
\providecommand{\BIBentrySTDinterwordspacing}{\spaceskip=0pt\relax}
\providecommand{\BIBentryALTinterwordstretchfactor}{4}
\providecommand{\BIBentryALTinterwordspacing}{\spaceskip=\fontdimen2\font plus
\BIBentryALTinterwordstretchfactor\fontdimen3\font minus
  \fontdimen4\font\relax}
\providecommand{\BIBforeignlanguage}[2]{{%
\expandafter\ifx\csname l@#1\endcsname\relax
\typeout{** WARNING: IEEEtran.bst: No hyphenation pattern has been}%
\typeout{** loaded for the language `#1'. Using the pattern for}%
\typeout{** the default language instead.}%
\else
\language=\csname l@#1\endcsname
\fi
#2}}
\providecommand{\BIBdecl}{\relax}
\BIBdecl

\bibitem{bpr}
S.~Basu, R.~Pollack, and M.-F. Roy, \emph{Algorithms in Real Algebraic
  Geometry}.\hskip 1em plus 0.5em minus 0.4em\relax Springer-Verlag, 2006.

\bibitem{benediktlibkinominimal}
M.~Benedikt and L.~Libkin, ``Relational queries over interpreted structures,''
  \emph{J. {ACM}}, vol.~47, no.~4, pp. 644--680, 2000.

\bibitem{benediktsurvey}
M.~Benedikt, ``Generalizing finite model theory,'' in \emph{Logic Colloquium
  '03}.\hskip 1em plus 0.5em minus 0.4em\relax Cambridge University Press,
  2006, p. 3–24.

\bibitem{libkinsurvey}
L.~Libkin, ``Embedded finite models and constraint databases,'' in \emph{Finite
  Model Theory and Its Applications}.\hskip 1em plus 0.5em minus 0.4em\relax
  Springer, 2007, pp. 257--337.

\bibitem{3belgians}
J.~Paredaens, J.~V. den Bussche, and D.~V. Gucht, ``First-order queries on
  finite structures over the reals,'' \emph{{SIAM} J. Comput.}, vol.~27, no.~6,
  pp. 1747--1763, 1998.

\bibitem{flumziegler}
J.~Flum and M.~Ziegler, ``Pseudo-finite homogeneity and saturation,''
  \emph{{Journal of Symbolic Logic}}, vol.~64, no.~4, pp. 1689--1699, 1999.

\bibitem{casanovasziegler}
E.~Casanovas and M.~Ziegler, ``Stable theories with a new predicate,''
  \emph{{Journal of Symbolic Logic}}, vol.~66, no.~3, 2001.

\bibitem{belegradekisolation}
O.~V. Belegradek, A.~P. Stolboushkin, and M.~A. Taitslin, ``Extended
  order-generic queries,'' \emph{{Annals of Pure and Applied Logic}}, vol.~97,
  no. 1-3, pp. 85--125, 1999.

\bibitem{baldben}
J.~Baldwin and M.~Benedikt, ``Stability theory, permutations of indiscernibles,
  and embedded finite models,'' \emph{{Transactions of the AMS}}, 2000.

\bibitem{ax}
J.~Ax, ``The elemenentary theory of finite fields,'' \emph{Annals of
  Mathematics}, vol.~88, no.~2, pp. 293--271, 1968.

\bibitem{zoelouangus}
Z.~Chatzidakis, L.~Van~den Dries, and A.~Macintyre, ``Definable sets over
  finite fields,'' \emph{Journal f\"ur die reine und angewandte Mathematik},
  vol. 427, pp. 107 -- 136, 1992.

\bibitem{simontwovarsip}
\BIBentryALTinterwordspacing
P.~Simon, ``A note on stability and {NIP} in one variable,'' 2021. [Online].
  Available: \url{https://arxiv.org/abs/2103.15799}
\BIBentrySTDinterwordspacing

\bibitem{simon-chernikov-two}
P.~Simon and A.~Chernikov, ``Externally definable sets and dependent pairs:
  {II},'' \emph{{Transactions of the AMS}}, 2015.

\bibitem{szymonnip}
\BIBentryALTinterwordspacing
J.~Dreier, N.~M{\"{a}}hlmann, S.~Siebertz, and S.~Torunczyk, ``Indiscernibles
  and flatness in monadically stable and monadically {NIP} classes,'' in
  \emph{{ICALP}}, 2023. [Online]. Available:
  \url{https://doi.org/10.4230/LIPIcs.ICALP.2023.125}
\BIBentrySTDinterwordspacing

\bibitem{henson72}
C.~W. Henson, ``Countable homogeneous relational structures and
  $\aleph_0$-categorical theories,'' \emph{{Journal of Symbolic Logic}},
  vol.~37, no.~3, 1972.

\bibitem{bodirskyhenson}
\BIBentryALTinterwordspacing
M.~Bodirsky, B.~Martin, M.~Pinsker, and A.~Pongr\'{a}cz, ``{Constraint
  Satisfaction Problems for Reducts of Homogeneous Graphs},'' \emph{SIAM
  Journal on Computing}, vol.~48, no.~4, pp. 1224--1264, 2019. [Online].
  Available: \url{https://doi.org/10.1137/16M1082974}
\BIBentrySTDinterwordspacing

\bibitem{definabilitypatterns}
E.~Hrushovski, ``Definability patterns and their symmetries,'' 2019,
  \url{https://arxiv.org/abs/1911.01129}.

\bibitem{hodgesbook}
W.~Hodges, \emph{Model Theory}.\hskip 1em plus 0.5em minus 0.4em\relax
  Cambridge University Press, 1993.

\bibitem{ominimal}
L.~P. D. v.~d. Dries, \emph{Tame Topology and O-minimal Structures}.\hskip 1em
  plus 0.5em minus 0.4em\relax Cambridge University Press, 1998.

\bibitem{simonnip}
P.~Simon, \emph{A Guide to NIP Theories}.\hskip 1em plus 0.5em minus
  0.4em\relax Cambridge University Press, 2015.

\bibitem{laskowski}
M.~C. Laskowski, ``Vapnik-chervonenkis classes of definable sets,'' \emph{J.
  London Math. Soc.}, vol.~45, no.~2, 1992.

\bibitem{hullsunatact}
R.~Hull and J.~Su, ``Domain independence and the relational calculus,''
  \emph{Acta Informatica}, vol.~31, no.~6, pp. 513--524, 1994.

\bibitem{nicoleorderinvariant}
N.~Schweikardt, ``A short tutorial on order-invariant first-order logic,'' in
  \emph{Computer Science -- Theory and Applications}, 2013.

\bibitem{immermanbook}
N.~Immerman, \emph{Descriptive Complexity}.\hskip 1em plus 0.5em minus
  0.4em\relax Springer, 1999.

\bibitem{hellahigherorder}
L.~Hella and J.~M. Turull-Torres, ``Computing queries with higher-order
  logics,'' \emph{Theor. Comput. Sci.}, vol. 355, no.~2, pp. 197--214, Apr.
  2006.

\bibitem{mostowskisemantics}
M.~Mostowski, ``On representing semantics in finite models,'' in
  \emph{Philosophical Dimensions of Logic and Science}, A.~Rojszczak,
  J.~Cachro, and G.~Kurczewski, Eds.\hskip 1em plus 0.5em minus 0.4em\relax
  Kluwer, 2003, pp. 15--28.

\bibitem{leszek}
L.~A. Kolodziejczyk, ``Truth definitions in finite models,'' \emph{J. Symb.
  Log.}, vol.~69, no.~1, pp. 183--200, 2004.

\bibitem{hullsuhigherorder}
R.~Hull and J.~Su, ``On the expressive power of database queries with
  intermediate types,'' \emph{J. Comput. Syst. Sci.}, vol.~43, no.~1, 1991.

\bibitem{thesisspectrum}
J.~H. Bennett, ``{On Spectra},'' Ph.D. dissertation, Princeton, 1962.

\bibitem{dugaldsurvey}
D.~Macpherson, ``A survey of homogeneous structures,'' \emph{Discrete
  Mathematics}, vol. 311, no.~15, pp. 1599--1634, 2011.

\bibitem{kozenbooleanalg}
D.~Kozen, ``{Complexity of Boolean algebras},'' \emph{Theoretical Computer
  Science}, vol.~10, no.~3, pp. 221--247, 1980.

\bibitem{kpt}
A.~S. K. V.~G. Pestov and S.~Todorcevic, ``{Fraısse limits, Ramsey theory, and
  topological dynamics of automorphism groups},'' \emph{Geom. Funct. Anal.},
  vol.~15, no.~1, p. 106–189, 2005.

\bibitem{fieldarith}
M.~Fried and M.~Jarden, \emph{Field arithmetic}.\hskip 1em plus 0.5em minus
  0.4em\relax Springer, 2008.

\bibitem{wagnersimple}
F.~Wagner, \emph{Simple theories}.\hskip 1em plus 0.5em minus 0.4em\relax
  Kluwer, 2000.

\bibitem{dugaldcharlie}
D.~Macpherson and C.~Steinhorn, ``One-dimensional asymptotic classes of finite
  structures,'' \emph{{Transactions of the AMS}}, vol. 360, no.~1, 2008.

\bibitem{FriedSacerdote}
M.~Fried and G.~Sacerdote, ``Solving {D}iophantine problems over all residue
  class fields of a number field and all finite fields,'' \emph{Ann. of Math.
  (2)}, vol. 104, no.~2, pp. 203--233, 1976.

\bibitem{emm}
E.~Hrushovski, ``The {Manin-Mumford} conjecture and the model theory of
  difference fields,'' \emph{{Annals of Pure and Applied Logic}}, vol. 112, pp.
  43--115, 2001.

\bibitem{garcia}
D.~Garcia, D.~Macpherson, and C.~Steinhorn, ``Pseudofinite structures and
  simplicity,'' \emph{Journal of Mathematical Logic}, vol.~15, no.~1, 2015.

\bibitem{zoenotes}
\BIBentryALTinterwordspacing
Z.~Chatazadakis, ``Notes on the model theory of finite and pseudo-finite
  fields,'' 2018. [Online]. Available:
  \url{https://ivv5hpp.uni-muenster.de/u/baysm/mtpsf/Madrid05-Course-PSF.pdf}
\BIBentrySTDinterwordspacing

\bibitem{distalwild}
P.~Hieronymi, T.~Nell, and E.~Walsberg, ``Wild theories with o-minimal open
  core,'' \emph{{Annals of Pure and Applied Logic}}, vol. 169, no.~2, pp.
  146--163, 2018.

\bibitem{simon-chernikov-one}
P.~Simon and A.~Chernikov, ``Externally definable sets and dependent pairs,''
  \emph{Israeli Journal of Mathematics}, 2013.

\bibitem{louremarkax}
\BIBentryALTinterwordspacing
L.~v.~d. Dries, ``A remark on ax's theorem on solvability modulo primes.''
  \emph{Mathematische Zeitschrift}, vol. 208, no.~1, pp. 65--70, 1991.
  [Online]. Available: \url{http://eudml.org/doc/174301}
\BIBentrySTDinterwordspacing

\bibitem{ch}
G.~Cherlin and E.~Hrushovski, \emph{Finite Structures with Few Types}, ser.
  Annals of Mathematics Studies.\hskip 1em plus 0.5em minus 0.4em\relax PUP,
  2003, no. 152.

\end{thebibliography}
